\documentclass{article}
\PassOptionsToPackage{table,xcdraw}{xcolor}
\usepackage{arxiv_style,times}
\usepackage{adjustbox}
\usepackage{algorithm}
\usepackage{algorithmic}
\usepackage{textcomp}
\usepackage{mathtools}
\usepackage{microtype}
\usepackage{graphicx}
\usepackage{caption}
\usepackage{booktabs}
\usepackage{hyperref}
\usepackage{xurl}
\usepackage{array}
\usepackage{amsmath,amsfonts}
\usepackage{amssymb}
\usepackage{amsthm}
\usepackage{enumitem}
\usepackage{cases}
\usepackage[capitalise]{cleveref}

\usepackage{footnotehyper}
\makesavenoteenv{minipage}

\newcommand{\Require}{\REQUIRE}
\newcommand{\State}{\STATE}
\newcommand{\For}[1]{\FOR{#1}}
\newcommand{\EndFor}{\ENDFOR}
\newcommand{\Return}{\textbf{return}\ }
\makeatletter\@ifundefined{algorithmicindent}{\newlength{\algorithmicindent}}{}\makeatother
\renewcommand{\algorithmiccomment}[1]{\hfill{\color{gray!70}$\triangleright$~#1}}
\makeatletter
\newcommand{\FORR}[1]{\ALC@it\algorithmicfor\ #1%
  \begin{ALC@loop}}
\newcommand{\ENDFORR}{\end{ALC@loop}\ALC@it\algorithmicendfor}
\makeatother

\newtheorem{proposition}{Proposition}

\newtheorem{lemma}{Lemma}
\newtheorem{theorem}{Theorem}
\newtheorem{corollary}{Corollary}

\definecolor{lightgraybox}{gray}{0.94}
\colorlet{metablue}{blue!60!green}
\colorlet{mygreen}{green!55!black}
\definecolor{myyellow1}{HTML}{D89A3C}
\colorlet{myyellow}{myyellow1!90!black}
\newcommand{\cellhi}{\cellcolor{metablue!15}}

\newcommand{\myellow}[1]{{\color{myyellow}#1}}
\newcommand{\mgray}[1]{{\text{ \hypersetup{hidelinks} \color{lightgray} #1}}}
\newcommand{\shortnote}[1]{{\tag*{\mgray{#1}}}}
\newcommand{\graysd}[1]{{\color{gray}\tiny$\pm\,#1$}}
\hypersetup{colorlinks,linkcolor=metablue,citecolor=metablue,urlcolor=metablue}

\newcommand{\graybox}[1]{%
    \par\noindent
    \colorbox{lightgraybox}{%
        \parbox{\dimexpr\linewidth-2\fboxsep\relax}{#1}%
    }%
    \par
}

\newcommand{\diff}{\mathop{}\!{\mathrm{d}}}
\newcommand{\tr}{\mathrm{tr}}
\newcommand{\T}{{\top}}
\newcommand{\R}{\mathbb{R}}

\DeclareMathOperator*{\argmin}{arg\,min}

\title{Hard-Constrained Sampling on Embedded Riemannian Manifolds via Adjoint Schr\"{o}dinger Bridges}

\author{%
\textbf{Mattia Mosso}$^{1}$ \qquad
\textbf{Jaemoo Choi}$^{1}$ \qquad
\textbf{Heng Yang}$^{2}$\\[0.45em]
{\small $^{1}$School of Aerospace Engineering, Georgia Institute of Technology}\\
{\small $^{2}$School of Engineering and Applied Sciences, Harvard University}\\[0.3em]
{\footnotesize\texttt{mattia\_mosso@seas.harvard.edu \quad jchoi843@gatech.edu \quad hankyang@seas.harvard.edu}}
}

\begin{document}

\maketitle
\addtocontents{toc}{\protect\setcounter{tocdepth}{0}}

\begin{abstract}
\noindent
A variety of tasks require sampling from unnormalized Boltzmann distributions supported on manifolds. Building upon the foundations of adjoint matching and adjoint Schr\"{o}dinger bridge sampling, this paper provides a theoretically justified method, through the lens of stochastic optimal control, to address this problem on smooth, compact, path-connected embedded Riemannian manifolds.
As an element of novelty compared to existing literature, feasibility is enforced at the level of the state space, meaning the controlled diffusion is defined intrinsically on the curved space. 
Empirical validations are provided for several physics applications.
\end{abstract}

\section{Introduction}
\label{sec:intro}
Sampling from unnormalized Gibbs--Boltzmann distributions is a fundamental problem in computational physics, molecular simulation, Bayesian inference, and statistical mechanics. Given an energy function $E:\R^d\to\R$, the objective is to generate samples from the distribution
\begin{equation}
    \nu(\mathrm{d}x)
    =
    \frac{1}{Z}\exp\bigl(-E(x)\bigr)\,\mathrm{d}x,
    \qquad
    Z
    :=
    \int_{\R^d}\exp\bigl(-E(x)\bigr)\,\mathrm{d}x,
    \label{eq: pi_Rd}
\end{equation}
without evaluating the generally intractable normalizing constant $Z$.
Classical Markov chain Monte Carlo methods (MCMC) \citep{metropolis1953equation, neal2001annealed, delmoral2006sequential} can mix slowly in high-dimensional or multimodal energy landscapes because each new sample is generated through a sequential transition from the previous state.
Diffusion neural samplers instead learn a finite-horizon stochastic process that transports a tractable source distribution to the target, thereby providing amortized sample generation after training.
Within this class, stochastic optimal control (SOC) frameworks have given rise to path integral sampler \citep{zhang2022path}, while regression-based SOC solvers such as adjoint matching (AM) underpin scalable solutions including adjoint sampling and adjoint Schr\"odinger bridge sampler (ASBS), without requiring samples from the target distribution
\citep{domingoenrich2025adjoint,havens2025adjoint,liu2025adjoint,Guo2026DiscreteASBS}.

However, existing adjoint-sampling formulations are developed primarily for Euclidean state spaces and do not directly enforce exact support constraints.
Many scientific and engineering variables instead belong to lower-dimensional geometric spaces.
Examples include unit-norm vectors, rigid-body rotations, orthogonal matrices, conservation-law level sets, and closed-loop robot configurations.
Orthogonality constraints, in particular, lead to Stiefel manifolds in Bayesian factor models, probabilistic principal component analysis, and low-rank matrix models
\citep{jauch2021bayesian,pourzanjani2021bayesian,cui2024lowrank}.
This motivates us to consider a smooth embedded manifold
\begin{equation}
    \mathcal{M}
    :=
    \left\{
        x\in\R^d
        :
        c(x)=0
    \right\},
    \qquad
    c:\R^d\to\R^m,
    \quad m<d,
    \label{eq: M_constraint}
\end{equation}
and the intrinsic target measure
\begin{equation}
    \nu_{\mathcal{M}}(\mathrm{d}x)
    =
    \frac{1}{Z_{\mathcal{M}}}
    \exp\bigl(-E(x)\bigr)
    \,d\mathrm{vol}_{\mathcal{M}}(x),
    \qquad x\in\mathcal{M}.
    \label{eq: pi_M}
\end{equation}
Because $\nu_{\mathcal{M}}$ is supported on a lower-dimensional subset of $\R^d$, it is singular with respect to ambient Lebesgue measure.
Consequently, a Euclidean controlled diffusion does not preserve feasibility, while penalty formulations generally provide only approximate constraint satisfaction.

\definecolor{colordm}{HTML}{01847F}
\definecolor{colortm}{HTML}{920872}

\newcommand{\colordm}[1]{{\color{colordm}#1}}
\newcommand{\colortm}[1]{{\color{colortm}#1}}

\begin{table*}[t]
\centering
\caption{
Conceptual comparison between Euclidean and embedded Riemannian state
spaces.
}
\label{tab:euclidean_manifold_comparison}
\resizebox{\linewidth}{!}{
\renewcommand{\arraystretch}{1.15}
\begin{tabular}{ccc}
\toprule
&
  Euclidean state space $\R^d$
&
  Embedded Riemannian state space
  $\mathcal M\subset\R^d$
\\
\midrule

Ref. dyn. $p^{\mathrm {base}}$
&
  $\diff X_t=f_t(X_t)\diff t+\sigma_t\diff W_t$,
  $X_0\sim\mu$
&
  $\diff X_t= f_t^{\mathcal M}(X_t)\diff t
  +\sigma_t\diff W_t^{\mathcal M}$,
  $X_0\sim\mu_{\mathcal M}$
\\

Ctrl. dyn. $p^u$
&
  $\diff X_t=\bigl(f_t+\sigma_tu_t\bigr)(X_t)\diff t
  +\sigma_t\diff W_t$,
  $X_0\sim\mu$
&
  $\diff X_t=
  \bigl(f_t^{\mathcal M}+\sigma_tu_t\bigr)(X_t)\diff t
  +\sigma_t\diff W_t^{\mathcal M}$,
  $X_0\sim\mu_{\mathcal M}$,
  $u_t(x)\in T_x\mathcal M$
\\
\midrule

SB Prob.
&
  $\min\limits_{u~\text{s.t.}~X_1\sim\nu}
  \underset{\boldsymbol X\sim p^u}{\mathbb E}
  \Big[
    \int\limits_0^1
    \frac12\|u_t(X_t)\|^2\diff t
  \Big]$
&
  $\min\limits_{u~\text{s.t.}~X_1\sim\nu_{\mathcal M}}
  \underset{\boldsymbol X\sim p_{\mathcal M}^u}{\mathbb E}
  \Big[
    \int\limits_0^1
    \frac12
    \|u_t(X_t)\|_{T_{X_t}\mathcal M}^2
    \diff t
  \Big]$
\\

SOC Prob.
&
  $\min\limits_u
  \underset{\boldsymbol X\sim p^u}{\mathbb E}
  \Big[
    \int\limits_0^1
    \frac12\|u_t(X_t)\|^2\diff t
    +
    \log\frac{\widehat\varphi_1}{\nu}(X_1)
  \Big]$
&
  $\min\limits_u
  \underset{\boldsymbol X\sim p_{\mathcal M}^u}{\mathbb E}
  \Big[
    \int\limits_0^1
    \frac12
    \|u_t(X_t)\|_{T_{X_t}\mathcal M}^2
    \diff t
    +
    \log
    \frac{\widehat\varphi_1^{\mathcal M}}
         {\nu_{\mathcal M}}(X_1)
  \Big]$
\\
\midrule

Opt. Ctrl.
&
  $u_t^\star(x)
  =
  \sigma_t\nabla\log\varphi_t(x)$
&
  $u_t^\star(x)
  =
  \sigma_t\nabla_{\mathcal M}
  \log\varphi_t^{\mathcal M}(x)
  \in T_x\mathcal M$
\\

Corrector
&
  $\nabla\log\widehat\varphi_1(x)$
&
  $\nabla_{\mathcal M}
  \log\widehat\varphi_1^{\mathcal M}(x)
  \in T_x\mathcal M$
\\
\midrule

\begin{tabular}[c]{@{}c@{}}
Zero-drift\\
isotropic reference
\end{tabular}
&
  \begin{tabular}[c]{@{}c@{}}
  $f_t\equiv0
  \Rightarrow
  p_{1|t}^{\mathrm{base}}(x_1|x)
  =
  q_t(x_1-x)$
  \\
  $q_t(\varepsilon)
  =
  \mathcal N
  \bigl(
    \varepsilon;
    0,\bar\sigma_t^2I_d
  \bigr),
  \
  \bar\sigma_t^2
  :=
  \int_t^1\sigma_s^2\diff s$
  \end{tabular}
&
  \begin{tabular}[c]{@{}c@{}}
  $f_t^{\mathcal M}\equiv0
  \Rightarrow
  p_{1|t}^{\mathrm{base},\mathcal M}(x_1|x) 
  =
  h_{\bar\sigma_t^2}^{\mathcal M}(x,x_1) \ \  \text{intrinsic heat kernel}$
  \\
  $h_{\bar\sigma_t^2}^{\mathcal M}(x,x_1)
  \neq
  q_t(x_1-x)
  \ \text{in general}$
  \end{tabular}
\\
\midrule

\begin{tabular}[c]{@{}c@{}}
\textbf{\colordm{Denoising}}\\
\textbf{\colordm{Matching}}
\end{tabular}
&
  \begin{tabular}[c]{@{}c@{}}
  $\nabla\log\widehat\varphi_1(x)
  =
  \mathbb E_{p_{0|1}^\star(x_0|x)}
  \!\left[
    \colordm{
    \nabla_x
    \log p_{1|0}^{\mathrm{base}}(x|x_0)}
  \right]$
  \end{tabular}
&
  \begin{tabular}[c]{@{}c@{}}
  $\nabla_{\mathcal M}
  \log\widehat\varphi_1^{\mathcal M}(x)
  =
  \mathbb E_{p_{0|1}^{\star,\mathcal M}(x_0|x)}
  \!\left[
    \colordm{
    \nabla_{\mathcal M,x}
    \log
    p_{1|0}^{\mathrm{base},\mathcal M}(x|x_0)}
  \right]$
  \end{tabular}
\\
\midrule

\begin{tabular}[c]{@{}c@{}}
\textbf{\colortm{Adjoint}}\\
\textbf{\colortm{Matching}}
\end{tabular}
&
  \begin{tabular}[c]{@{}c@{}}
  $\nabla\log\varphi_t(x)
  =
  \mathbb E_{p_{1|t}^\star(x_1|x)}
  \!\left[
    \colortm{
    \nabla\log\varphi_1(x_1)}
  \right]$
  \end{tabular}
&
  \begin{tabular}[c]{@{}c@{}}
  $\nabla_{\mathcal M}
  \log\varphi_t^{\mathcal M}(x)
  =
  \mathbb E_{p_{\mathcal M}^{\star}(\cdot\mid X_t=x)}
  \!\left[
    (\mathcal W_{t,1})^{*}
    \colortm{
    \nabla_{\mathcal M}
    \log\varphi_1^{\mathcal M}(X_1)}
  \right]$
  \end{tabular}
\\
\bottomrule
\end{tabular}
}
\end{table*}

Sampling directly on manifolds has traditionally been addressed using constrained or Riemannian variants of Langevin and Hamiltonian Monte Carlo methods \citep{brubaker2012family,girolami2011riemann,cheng2022efficient}.
These methodologies are geometrically principled and can provide asymptotically exact samples under appropriate assumptions. Nevertheless, they remain sequential Markov-chain procedures and can therefore exhibit slow mixing.
For implicitly defined nonlinear manifolds, they may also require repeated constraint solves, geometric operations, and careful discretization at every sampling step.
Additional related work is discussed in Appendix~\ref{app: related_work}.

In this paper, we extend ASBS to unnormalized distributions supported on smooth, compact, path-connected embedded Riemannian manifolds.
We formulate sampling as a Schr\"odinger bridge problem relative to an intrinsic manifold diffusion and establish its equivalent terminal-cost SOC formulation.
The resulting controlled process evolves directly on $\mathcal{M}$, with both its drift and diffusion acting through the tangent bundle $T\mathcal{M}$, so feasibility is enforced at the level of the state space.

The main difficulty is that Euclidean ASBS relies on additive Gaussian transition kernels and on the canonical identification of all tangent spaces with $\R^d$.
Neither property holds on a curved manifold (see Table~\ref{tab:euclidean_manifold_comparison}).
We replace these Euclidean operations with intrinsic transition kernels, manifold bridge constructions, and vector transport between tangent spaces.
This leads to Riemannian ASBS (R--ASBS), which uses local heat-kernel and geodesic approximations together with Log maps and parallel transport.
For general embedded manifolds where these geometric primitives are unavailable or expensive, we introduce Extended R--ASBS.
Both algorithms alternate between two supervised regression problems: controller matching transports terminal adjoints to intermediate states, while corrector matching removes the endpoint bias induced by a general non-memoryless source distribution.
The resulting architecture avoids differentiating through the simulated stochastic trajectories and supports batch-parallel training and amortized sample generation.

\paragraph{Contributions.}
Our main contributions are summarized as follows:
\begin{enumerate}[label=\textbf{\Roman*.}]
    \item \textbf{Sampling, Schr\"odinger bridges, and stochastic optimal control on manifolds.}
    We establish the intrinsic SB--SOC equivalence and reciprocal property on \eqref{eq: M_constraint}.
    \item \textbf{Riemannian adjoint Schr\"odinger bridge sampling.}
    We derive exact denoising- and adjoint-matching identities, and obtain R--ASBS by replacing their generally intractable heat-kernel, bridge, and damped-transport terms with practical geometric approximations. 
    \item \textbf{An implementable extension for implicit manifolds.}
    We introduce Extended R--ASBS, which replaces unavailable Log maps, parallel transport, and heat-kernel scores with nearest-point retraction, projection-as-transport, and a projected-chord corrector.
    \item \textbf{Empirical validation.}
    We evaluate the proposed methods on multimodal spherical distributions, Gibbs sampling on Stiefel manifolds, high-dimensional closed-loop inverse kinematics with obstacle avoidance, and robust Wahba optimization on $\mathrm{SO}(3)$.
\end{enumerate}

\section{Mathematical Preliminaries} \label{sec: preliminaries}
\textbf{Notation.}
Let $\R^d$ denote the ambient Euclidean space and let $\langle \cdot, \cdot \rangle$ and $\| \cdot \|$ denote the standard inner product and norm. For a smooth scalar function $f: \R^d \to \R$, we write $\nabla_\alpha f$ and $\nabla_\alpha^2 f$ for its Euclidean gradient and Hessian with respect to the variables $\alpha$. Let $(\Omega, \mathcal{F}, \{ \mathcal{F}_t \}_{t \geq 0}, \mathbb{P})$ be a filtered probability space satisfying the usual conditions.

In the following section, we briefly summarize the main results of \citet{liu2025adjoint}, which will serve as a foundational background for our proposed method.

\subsection{Optimizing Diffusion Processes for Unconstrained Sampling}
\label{subsec: SOC_formulation}
Let $\nu$ be defined as in \eqref{eq: pi_Rd} and let $\mu$ be an easy-to-sample initial source distribution on $\R^d$. A diffusion sampler is specified by a base drift
$f_t(x): [0,1] \times \R^d \to \R^d$, a noise schedule $\sigma_t : [0,1] \to \R_{>0}$, and a learned control field $u_t^\theta (x)$, parameterized by $\theta$, transporting samples to the target distribution $\nu(x)$ at the terminal time $t=1$.
The controlled diffusion process is described by the control-affine It\^{o} stochastic differential equation (SDE) \citep{sarkka2019applied}:
\begin{equation}
    \diff X_t = [f_t (X_t) + \sigma_t u_t^\theta (X_t)] \diff t + \sigma_t \, \diff W_t, \qquad 
    X_0 \sim \mu,
    \label{eq: flat_controlled_sde}
\end{equation}
where $(W_t)_{t \in [0,1]}$ is a standard Brownian motion in $\R^d$.
Define the full sample trajectory as $\boldsymbol{X} := \{X_t: t \in [0,1]\}$. The probability law of this random trajectory induced by \eqref{eq: flat_controlled_sde} on the trajectory space $C([0,1], \R^d)$ is called the controlled path measure and is denoted by $p^u$. 
Similarly, $p^{\mathrm{base}}$ denotes the path measure corresponding to the uncontrolled dynamics obtained by setting $u \equiv 0$.

\subsection{Euclidean Adjoint Schr\"{o}dinger Bridge Sampler}
Euclidean ASBS formulates transport from the source $\mu$ to the target $\nu$ as a Schr\"odinger bridge and uses its equivalent terminal-cost stochastic optimal control formulation. Denoting the associated SB potentials by $\varphi_t$ and $\widehat\varphi_t$, the optimal control satisfies $u_t^\star=\sigma_t\nabla\log\varphi_t$. The complete SB--SOC derivation and the corresponding denoising- and adjoint-matching identities are given in Appendix~\ref{sec: euclidean_SB_SOC}.

Specializing to a Brownian-motion base process ($f_t\equiv0$), adjoint matching\footnote{A review on adjoint matching and reciprocal adjoint matching is provided in Appendix~\ref{app: adjoint}.} gives the ASBS \emph{adjoint-matching} condition
\graybox{%
\begin{equation}
    u^\star = \argmin_{u \in \mathcal{U}} \ \mathbb{E}_{p_{t \vert 0,1}^{\mathrm{base}} p_{0,1}^{\bar{u}}} \big[ \|u_t(X_t) + \sigma_t (\nabla E + \nabla \log \widehat{\varphi}_1) (X_1)\|^2 \big], \qquad t\sim\mathrm{Unif}[0,1]
\label{eq: ASBS_u_star}
\end{equation}
}
where $\bar{u} = \texttt{stopgrad}(u)$, $p^{\mathrm{base}}_{t|0,1}$ denotes the base-process bridge conditional law of $X_t$ given the endpoints $(X_0,X_1)$, and $p^{\bar u}_{0,1}$ is the current controlled endpoint law. Thus, as in reciprocal adjoint matching (RAM), the regression can be performed using endpoint samples from the current controlled process and intermediate points sampled from the base bridge\footnote{Strictly speaking, \eqref{eq: ASBS_u_star} is a fixed-point characterization: the displayed minimizer equals $u^\star$ exactly when the current endpoint law satisfies $p_{0,1}^{\bar u}=p_{0,1}^{\star}$, as happens at convergence of the alternating scheme described below.}. 

Clearly, computing the AM objective \eqref{eq: ASBS_u_star} requires knowing $\nabla \log \widehat{\varphi}_1 (x)$, which serves as a corrector that debiases the optimization towards the desired target. Notably, by denoise matching, this corrector function also admits a variational form
\graybox{%
\begin{equation}
    \nabla \log \widehat{\varphi}_1 = \argmin_h \ \mathbb{E}_{p_{0,1}^{u^\star}} \big[ \| h(X_1) - \nabla_{X_1} \log p^{\mathrm{base}} (X_1 \vert X_0) \|^2 \big].
    \label{eq: euclidean_corrector}
\end{equation}
}
This objective is called \emph{corrector matching}. It estimates the endpoint corrector using samples from the current controlled process, without requiring samples from the target distribution.
In the context of Euclidean ASBS corrector, computing \eqref{eq: euclidean_corrector} is straightforward, since it only requires evaluating this transition from time $0$ to $1$, leading to the linear score $\nabla_{X_1} \log p_{1|0}^{\mathrm{base}}(X_1 \vert X_0) = - \frac{X_1 - X_0}{\int_0^1 \sigma_t^2 \diff t}$.

To summarize, \eqref{eq: ASBS_u_star} and \eqref{eq: euclidean_corrector} characterize two distinct, but interdependent, matching objectives that any kinetic-optimal drift $u_t^\star$ of SBs must satisfy.
ASBS relaxes the interdependency with an alternating optimization scheme. 
Intuitively, at each stage $k$, we first find the control $u^{(k)}$ that best aligns with the corrector from the previous stage, $h^{(k-1)}$, then update the corrector $h^{(k)}$ accordingly to reflect the memorylessness of the current control.
This alternating procedure can be interpreted as an iterative proportional fitting scheme between forward and backward half-bridges, and converges to the Schr\"{o}dinger bridge solution when each matching stage reaches its critical point, i.e., $\lim_{k \to \infty} u^{(k)} = u^\star$ \citep[Thm. 3.2]{liu2025adjoint}.

\section{Extension to Embedded Riemannian Constraint Manifolds} \label{sec: manifold_extension}
Throughout this section, we extend the Euclidean formulation reviewed above to hard-constrained sampling problems whose feasible set is \eqref{eq: M_constraint}. The aim is to construct a controlled diffusion evolving directly on $\mathcal{M}$, so that its terminal marginal coincides with $\nu_\mathcal{M}$.
In contrast to penalty-based or projected ambient approaches, feasibility is enforced at the level of the state space: the process is defined intrinsically on $\mathcal{M}$, and therefore remains feasible for all times.

\subsection{Reference and Controlled Diffusions on $\mathcal{M}$} \label{subsec: notation}
In the following, we assume $c\in C^\infty(\mathbb R^d,\mathbb R^m)$ and the Jacobian $J_c(x) \in \R^{m \times d}$ has full row rank $\forall x \in \mathcal{M}$, implying $\mathcal{M}$ is a smooth, boundaryless, embedded submanifold of dimension $n=d-m$. Furthermore, we assume that $\mathcal{M}$ is compact and path-connected.
We denote by $T_x \mathcal{M}$ the tangent space of $\mathcal{M}$ at $x$, and by $P_x := I_d - J_c(x)^\T (J_c J_c^\T)^{-1} J_c(x)$ the orthogonal projection from $\R^d$ onto $T_x \mathcal{M}$.
The Riemannian gradient, divergence, and Laplace--Beltrami operator on $\mathcal{M}$ are denoted by $\nabla_{\mathcal{M}}$, $\operatorname{div}_{\mathcal{M}}$, and $\Delta_{\mathcal{M}}$, respectively.
We denote by $W_t^{\mathcal{M}}$ Brownian motion on $\mathcal{M}$, namely the diffusion whose infinitesimal generator is $\frac12 \Delta_{\mathcal{M}}$.
Unless otherwise stated, all probability densities on $\mathcal{M}$ are understood with respect to $d\mathrm{vol}_{\mathcal{M}}$. We refer to \citet{lee2018introduction}, together with Appendix~\ref{app: geometric_stochastic_preliminaries}, for a detailed discussion on Riemannian manifolds.

Let $\mu_{\mathcal M}$ be a tractable source distribution on $\mathcal{M}$, absolutely continuous with respect to $d\mathrm{vol}_{\mathcal M}$. 
We consider a time-dependent tangent base drift, smooth in \((t,x)\),
\begin{equation*}
    f_t^{\mathcal M}:\mathcal M\to T\mathcal M,
    \qquad
    f_t^{\mathcal M}(x)\in T_x\mathcal M,
\end{equation*}
and a scalar noise schedule \(\sigma_t \in C^1([0,1])\) satisfying \(0<\underline{\sigma}\leq \sigma_t\leq\overline{\sigma}\). The reference process on $\mathcal{M}$ is the intrinsic diffusion with generator
\begin{equation}
    \mathcal L_t^{\mathrm{base},\mathcal M} \psi = \left\langle f_t^{\mathcal{M}},\nabla_{\mathcal M}\psi \right\rangle + \frac{\sigma_t^2}{2}\Delta_{\mathcal M}\psi,
    \label{eq: manifold_base_generator}
\end{equation}
for every smooth test function $\psi:\mathcal M\to \R$. Equivalently, we write\footnote{The notation \eqref{eq: manifold_base_sde} is understood through the generator \eqref{eq: manifold_base_generator}.}
\begin{equation}
    \diff X_t = f_t^{\mathcal M}(X_t) \diff t + \sigma_t\, \diff W_t^{\mathcal M}, \quad X_0\sim\mu_{\mathcal M}.
    \label{eq: manifold_base_sde}
\end{equation}

A controlled manifold diffusion is obtained by adding a tangent control field $u_t:\mathcal M\to T\mathcal M, \ u_t(x)\in T_x\mathcal M$, according to
\begin{equation}
    \diff X_t = \left[ f_t^{\mathcal M}(X_t) + \sigma_t u_t(X_t) \right] \diff t + \sigma_t\, \diff W_t^{\mathcal M},
    \label{eq: manifold_controlled_sde}
\end{equation}
with the corresponding generator
\begin{equation}
    \mathcal L_t^{u,\mathcal M}\psi = \left\langle f_t^{\mathcal M}+\sigma_tu_t, \nabla_{\mathcal M}\psi \right\rangle + \frac{\sigma_t^2}{2}\Delta_{\mathcal M}\psi.
    \label{eq: manifold_controlled_generator}
\end{equation}
Since \eqref{eq: manifold_controlled_sde} is defined intrinsically on $\mathcal M$, the hard constraint is preserved by construction (i.e., $X_0\in\mathcal M \implies X_t \in \mathcal M,  \ \forall t \in [0,1], \ p_{\mathcal{M}}^u \text{-a.s.}$).

\subsection{Manifold SOC--SB Equivalence}
Let $\mathcal U_{\mathcal M}$ denote the class of progressively measurable tangent controls $u_t(X_t)\in T_{X_t}\mathcal M$ for which the controlled martingale problem is well posed, the Girsanov density with respect to $p_{\mathcal M}^{\mathrm{base}}$ is a true martingale, and
$\mathbb E_{p_{\mathcal M}^{u}}\int_0^1\|u_t(X_t)\|_{T_{X_t}\mathcal M}^2\diff t<\infty$.
The preceding discussion yields the following counterpart of the ASBS SOC--SB equivalence (the related proofs for this section, with additional remarks, are presented in Appendix~\ref{app: manifold_SOC_formulation}).

\begin{theorem}[Embedded-manifold SOC--SB equivalence] \label{theo: SOC_SB_manifold_equivalence}
Let $\mathcal M\subset\mathbb R^d$ satisfy the hypotheses established in Section~\ref{subsec: notation}. Assume $\mu_{\mathcal M},\nu_{\mathcal M}\in C^2(\mathcal M)$ are strictly positive probability densities, $E\in C^2(\mathcal M)$, and the reference transition kernels admit smooth strictly positive densities with respect to $d\mathrm{vol}_{\mathcal M}$. 
Then the manifold Schr\"odinger Bridge problem
\begin{equation}
    \min_p D_{\mathrm{KL}} \left( p\,\|\,p_{\mathcal M}^{\mathrm{base}} \right), \qquad p_0=\mu_{\mathcal M}, \quad p_1=\nu_{\mathcal M},
    \label{eq: manifold_SB_problem}
\end{equation}
admits a unique solution $p_{\mathcal M}^\star$. There exist strictly positive Schr\"odinger potentials $\varphi_t^{\mathcal M},\widehat\varphi_t^{\mathcal M}$, unique up to reciprocal multiplicative constants, satisfying
\begin{align}
\varphi_t^{\mathcal M}(x)
&=\int_{\mathcal M}p_{1|t}^{\mathrm{base},\mathcal M}(y|x)\varphi_1^{\mathcal M}(y)\,d\mathrm{vol}_{\mathcal M}(y),
\label{eq: manifold_phi_backward}\\
\widehat\varphi_t^{\mathcal M}(x)
&=\int_{\mathcal M}p_{t|0}^{\mathrm{base},\mathcal M}(x|y)\widehat\varphi_0^{\mathcal M}(y)\,d\mathrm{vol}_{\mathcal M}(y),
\label{eq: manifold_phihat_forward}
\end{align}
with $\varphi_0^{\mathcal M}\widehat\varphi_0^{\mathcal M}=\mu_{\mathcal M}$ and $\varphi_1^{\mathcal M}\widehat\varphi_1^{\mathcal M}=\nu_{\mathcal M}$. Moreover,
\begin{equation}
\frac{d p_{\mathcal M}^\star}{d p_{\mathcal M}^{\mathrm{base}}}(\boldsymbol X)
=
\frac{\varphi_1^{\mathcal M}(X_1)}{\varphi_0^{\mathcal M}(X_0)}
=
\frac{\widehat\varphi_0^{\mathcal M}(X_0)}{\mu_{\mathcal M}(X_0)}\varphi_1^{\mathcal M}(X_1).
\label{eq: manifold_SB_RN}
\end{equation}
The optimal path measure is induced by \eqref{eq: manifold_controlled_sde} with
\graybox{%
\begin{equation}
    u_t^\star(x)=\sigma_t\nabla_{\mathcal M}\log\varphi_t^{\mathcal M}(x).
    \label{eq: manifold_optimal_control}
\end{equation}
}
It is also the unique path-law solution of the terminal-cost SOC problem
\graybox{%
\begin{equation}
    \min_{u\in\mathcal U_{\mathcal M}} \mathbb E_{\boldsymbol X\sim p_{\mathcal M}^{u}} \left[ \int_0^1 \frac12 \|u_t(X_t)\|_{T_{X_t}\mathcal{M}}^2 \diff t + \log \frac{\widehat\varphi_1^{\mathcal M}(X_1)}{\nu_{\mathcal M}(X_1)}
    \right].
    \label{eq: SOC_theo1}
\end{equation}
}
If $\nu_{\mathcal M}(x)\propto e^{-E(x)}$ as in \eqref{eq: pi_M}, its terminal adjoint is
\graybox{%
\begin{equation}
    \nabla_{\mathcal M} \log \frac{\widehat\varphi_1^{\mathcal M}(x)}{\nu_{\mathcal M}(x)} = \nabla_{\mathcal M}E(x) + \nabla_{\mathcal M} \log \widehat\varphi_1^{\mathcal M}(x).
    \label{eq: manifold_terminal_adjoint}
\end{equation}
}
\end{theorem}
Thus, as in Euclidean ASBS, the term $\nabla_{\mathcal M} \log \widehat\varphi_1^{\mathcal M}$ acts as a corrector that removes the initial value-function bias induced by a general non-memoryless source distribution. The difference is that the corrector is now an intrinsic tangent vector field on $\mathcal{M}$.

\subsection{Exact Intrinsic Matching Identities}
\label{subsec: exact_manifold_matching}
We next derive the exact matching identities from the manifold SB structure established in Theorem~\ref{theo: SOC_SB_manifold_equivalence}, and subsequently distinguish them from the geometric approximations used for computation. 

\colordm{Denoising matching (corrector)}. 
The denoising identity extends directly, yielding a verbatim equivalent on $\mathcal{M}$ of the Euclidean case \eqref{eq: euclidean_DM_identity}, because it only requires differentiating a smooth positive transition density; see Appendix~\ref{app: manifold_SOC_formulation}, Lemma~\ref{prop: exact_manifold_DM}.

\colortm{Adjoint matching (controller)}.
Unlike denoising matching, adjoint matching necessitates differentiating the diffusion semigroup itself. For the Brownian reference used by ASBS, the Euclidean identity $\nabla P_tf=P_t\nabla f$ is replaced on a curved manifold by a derivative-semigroup formula involving damped stochastic parallel transport \citep{thalmaier1998remarks,coulibaly2011brownian}.
For $f_t^{\mathcal M}\equiv0$, let $\mathcal W_{t,s}:T_{X_t}\mathcal M\to T_{X_s}\mathcal M$ solve, along the reference Brownian path,
\begin{equation}
\frac{D}{ds}\mathcal W_{t,s}v
=-\frac{\sigma_s^2}{2}\operatorname{Ric}_{X_s}^{\sharp}\!\left(\mathcal W_{t,s}v\right),
\qquad
\mathcal W_{t,t}=I,
\label{eq: damped_transport}
\end{equation}
where $D/ds$ is covariant differentiation along $X_s$ and $\operatorname{Ric}^{\sharp}$ is the Ricci endomorphism. Its metric adjoint $(\mathcal W_{t,s})^*:T_{X_s}\mathcal M\to T_{X_t}\mathcal M$ pulls back terminal covectors, identified with vectors through the metric, back to time $t$; see Appendix~\ref{app: manifold_SOC_formulation}, Lemma~\ref{prop: exact_manifold_AM}.

\subsection{Practical Geometric Approximations and R--ASBS}
\label{subsec: practical_geometric_approximations}
While Theorem~\ref{theo: SOC_SB_manifold_equivalence} and Lemma~\ref{prop: exact_manifold_DM}--\ref{prop: exact_manifold_AM} establish the exact continuous-time matching identities, they do not directly yield an implementable learning procedure.
A necessary next step is translating the exact intrinsic quantities to computable geometric surrogates.
In particular, Euclidean ASBS depends on three core operations: evaluating the base transition score used in corrector matching, sampling intermediate bridge states conditioned on the endpoints, and propagating endpoint adjoints to intermediate times. In $\R^d$, these admit closed-form expressions because Brownian transition densities are Gaussian and the derivative flow is the identity. On a manifold, the corresponding exact objects are the heat-kernel score, the reference diffusion bridge, and damped stochastic parallel transport. Algorithm~\ref{alg: asbs_m} approximates them, respectively, by a short-time heat-kernel score, a noisy geodesic bridge, and Levi--Civita parallel transport along the sampled bridge.

An analytical heat kernel solution, even when available, is often represented by an infinite series and is computationally prohibitive inside every neural-network update. We therefore approximate its score by the short-time Varadhan expansion \citep[\S\S5.1--5.2, 5.5]{hsu2002stochastic}
\begin{equation}
\nabla_{\mathcal M,X_1}\log p_{1|0}^{\mathrm{base},\mathcal M}(X_1|X_0)
\approx
\frac{\operatorname{Exp}_{X_1}^{-1}(X_0)}{\int_0^1\sigma_t^2\diff t}
-\frac12\nabla_{\mathcal M,X_1}\log\Theta(X_0,X_1),
\label{eq: Varadhan_approximation}
\end{equation}
where $\operatorname{Exp}_y^{-1}(x)\in T_y\mathcal M$ is the Log map and $\Theta$ is the Jacobian of the exponential map. The first term is the dominant geodesic correction and the second accounts for local volume distortion. Intermediate reference-bridge sampling is likewise approximated by injecting Euclidean bridge noise into the tangent space of the deterministic geodesic.

Finally, the exact adjoint target in \eqref{eq: manifold_exact_AM_loss} depends on the path-dependent damped operator $(\mathcal W_{t,1})^*$. Given a sampled approximate bridge path $\gamma_{t,1}$, we consider
\begin{equation}
(\mathcal W_{t,1})^*a_1
\approx
\mathcal T_{1\to t}^{\gamma}a_1,
\label{eq: practical_transport_approximation}
\end{equation}
where $\mathcal T_{1\to t}^{\gamma}:T_{X_1}\mathcal M\to T_{X_t}\mathcal M$ is Levi--Civita parallel transport along $\gamma$. Specifically, given a smooth curve $\gamma:[t,1]\to\mathcal M$, it is computed by
\begin{equation}
\nabla_{\dot\gamma_s}A_s=0,
\qquad
A_1=a_1\in T_{X_1}\mathcal M.
\label{eq: parallel_transport_equation}
\end{equation}
Approximation \eqref{eq: practical_transport_approximation} removes the Ricci damping in \eqref{eq: damped_transport} and replaces the stochastic reference path by the sampled geometric bridge. The endpoint-locality result in Appendix~\ref{app: theoretical_analysis} provides short-time theoretical support for these local heat-kernel, bridge, and transport substitutions, but does not constitute an exact convergence bound for the resulting algorithm.

This yields the proposed \emph{Riemannian Adjoint Schr\"odinger Bridge Sampler} (R--ASBS), shown in Algorithm~\ref{alg: asbs_m}. Further analyses on the algorithms' structure and complexity can be found in Appendix~\ref{app: additional_theoretical_results}.
Because the algorithm relies on closed-form Log maps and parallel transport together with the approximations \eqref{eq: Varadhan_approximation} and \eqref{eq: practical_transport_approximation}, it is designed for manifolds including $\mathbb S^n$, $\mathbb T^n$, and $\mathrm{SO}(3)$ for small $n$. The exponential map may be replaced by a more efficient (second-order) retraction $R_x: T_x\mathcal{M} \to \mathcal{M}$.

\subsection{Extension to General Embedded Riemannian Constraint Manifolds} \label{subsec: general_manifold_extension}
\emph{Extended R--ASBS} (Algorithm~\ref{alg: extended_asbs_m}) further replaces the parallel-transport surrogate in \eqref{eq: practical_transport_approximation} by Projection-as-Transport (PAT). 
That is, letting $t_j=j\Delta t$ and $X_j:=X_{t_j}$, a discretized path is written as $X_0,X_1,\dots,X_N$, backward parallel transport is approximated by sequential projection
\begin{equation}
    \mathcal{T}_{1 \to t_0}^\gamma (v) \approx P_{X_{t_0}} P_{X_{t_1}} \cdots P_{X_{t_{N-1}}} v \qquad \mathrm{with} \quad X_N=X_{t_N}=X_1,
    \label{eq: PAT}
\end{equation}
starting from $v = a_1(X_1) \in T_{X_1} \mathcal{M}$, with $a_1$ as in Lemma~\ref{prop: exact_manifold_AM}, and walking backward.
A further simplification is substituting the Log map in \eqref{eq: Varadhan_approximation} with a Projected chord corrector, which is obtained by replacing the exact corrector target with $\nabla_{\mathcal M, X_1} \log p_{1|0}^{\mathrm{base},\mathcal M}(X_1 \vert X_0) \approx - P_{X_1} \frac{X_1 - X_0}{\int_0^1 \sigma_t^2 \diff t}$, i.e., the projected Euclidean corrector \eqref{eq: euclidean_corrector}. 
To allow more generality, the retraction map $\Pi_\mathcal{M}$ is chosen to be the nearest-point projection onto $\mathcal{M}$, available via Newton's method\footnote{The map
$R_x^{\mathrm{proj}}(v):=\Pi_{\mathcal M}(x+v)$ in Algorithm \ref{alg: extended_asbs_m} defines a projection-based retraction in a neighborhood of
$\mathcal M$ \citep{absil2008optimization}.} on $c(x) = 0$.
Thus Extended R--ASBS introduces a second approximation layer: PAT replaces the ordinary-parallel-transport surrogate, while the projected chord replaces the heat-kernel denoising target. Notably, Algorithm~\ref{alg: extended_asbs_m} extends the full AM implementation, not RAM, because a computable reference bridge is unavailable in these more complex geometries.

\section{Numerical Experiments}
\label{sec: numerical_resuls}
\subsection{Spherical and Semi-Spherical Distributions for Earth and Climate Events}
\label{subsec: sphere_climate}
\begin{figure}[!t]
\centering
    \includegraphics[width=1.0\textwidth]{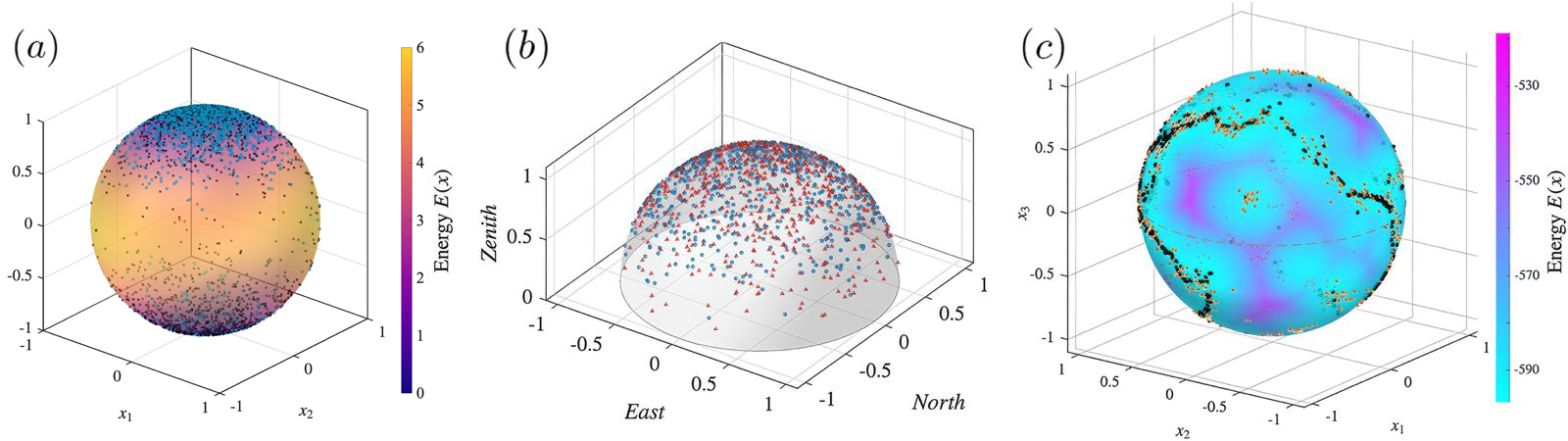}
    \caption{All cases utilize $\mu_\mathcal{M}(x) = 1/4\pi$. $(a)$ The generated samples ($\bullet$) effectively sit at the minimum energy regions, while geometric Langevin MCMC ($\textcolor{cyan}{\bullet}$) stays biased to the initial start $\delta_{(0,0,1)}$. $(b)$ shows the cosmic ray arrival directions (sky projection). The synthetic data ($\textcolor{red}{\blacktriangle}$) obtained with R--ASBS exhibit fidelity to the test data ($\textcolor{cyan}{\bullet}$), recovering the sparse distribution with Zenithal clustering. Finally, $(c)$ illustrates how the sampled points ($\textcolor{orange}{\blacktriangle}$) accumulate near the real data $\mathcal{D}_V$ ($\bullet$).
    }
    \label{fig: sphere_semisphere}
\end{figure}
A natural and geometrically intuitive application of R--ASBS lies in the unit sphere $\mathbb{S}^2 \subset \R^3$. 
This manifold finds direct employment in Earth and climate sciences, where events are intrinsically distributed on spherical surfaces (see benchmark of \citet{thornton2022riemannian}).
In particular, Algorithm~\ref{alg: asbs_m} is useful for evaluating probabilistic models of natural phenomena, such as earthquakes, floods, cosmic ray arrivals, and wildfires, by sampling from the model and directly comparing the generated data with real-world observations.

We first evaluate R--ASBS on a known analytical distribution with energy $E(x_1, x_2, x_3) := 6 (1-x_3^2)$, comparing it against a geometric Langevin MCMC baseline \citep{cheng2022efficient}. As shown in Figure~\ref{fig: sphere_semisphere}-$(a)$, the MCMC sampler exhibits characteristic slow mixing and initialization bias over a finite number of iterations. Due to the symmetry of $E$, the target assigns equal mass to the northern and southern hemispheres. Instead, $86\%$ of the MCMC particles become trapped in a single energy basin, failing to cross the high-energy equatorial region. In contrast, R--ASBS allocates $43.8\%$ of its particles to the northern hemisphere, close to the theoretical ratio of $0.5$.

Having established consistency of the sampler, we proceed to validate two empirical physical models: $(i)$ global earthquake distribution on $\mathbb{S}^2$; $(ii)$ cosmic ray arrivals to the hemisphere\footnote{The hemisphere is treated as a sphere, where reflection $x_3 \gets |x_3|$ is applied.} $\mathbb{S}_+^2 = \{x \in \mathbb{S}^2: x_3 \geq 0 \}$. Let $\mathcal{D}, \mathcal{D}_T, \mathcal{D}_V$ denote respectively the full dataset\footnote{The cosmic-ray dataset $SD1500$ is taken from \citet{PierreAugerOpenData}. For the earthquake dataset, we considered worldwide seismic events with magnitude $M \geq 6.5$ in the time window $1526-2026$, from \citet{USGSComCat}.}, training data and validation data.
To cast empirical evidence into a Boltzmann target $\nu_\mathcal{M} (x) \propto e^{- E_{\mathcal{D}_T} (x)}$, we define the energy landscape using a von Mises-Fisher kernel density estimator $E_{\mathcal{D}_T} (x) := - \log \Big( \frac{1}{N_{\text{train}}} \sum_{j=1}^{N_{\text{train}}} \exp(\kappa \langle x, z_j \rangle) \Big)$, where $\kappa$ acts as a smoothing parameter. 
Results are shown in Figure~\ref{fig: sphere_semisphere}-$(b),(c)$.

\subsection{Stiefel Manifold and Matrix Orthogonality Condition} \label{subsec: estimation_expectation}
In random matrix theory and quantum mechanics, elements of the Stiefel manifold (Appendix~\ref{app: app_stiefel_manifold}) represent an orthogonal basis of $p$ wavefunctions embedded in an $n$-dimensional Hilbert space. For a fixed symmetric matrix $H \in \R^{n \times n}$ representing the system Hamiltonian, the total energy of the state $X$ is given by the trace of the projected Hamiltonian $E(X) := \tr (X^\T H X)$.

When coupled to a thermal reservoir at temperature $T_r$, the states distribute according to the canonical Gibbs-Boltzmann distribution on the manifold
\begin{equation}
    \nu_{St(n,p)}(X) = \frac{1}{Z} \exp \left( -\frac{E(X)}{k_B T_r} \right)
    =  \frac{1}{Z} \exp \left( -\beta \ \tr (X^\T H X) \right), \qquad \beta := 1/(k_B T_r).
\label{eq: matrix_distribution}
\end{equation}
\noindent
\begin{minipage}[t]{0.56\linewidth}
\vspace{0pt}
The Euclidean gradient is $\nabla_X E(X) = 2 H X$, while the gradient of the negative log-density is $2 \beta H X$.

In the high-temperature regime $T_r\to\infty$ ($\beta\to0$), the target distribution converges to the uniform invariant probability measure on $St(n,p)$.
By its isotropy, the energy expectation is exactly $\mathbb{E}_{X \sim \mathrm{Unif}} [\tr(X^\T H X)] = \frac{p}{n} \tr H$.
If $N_s$ denotes the number of matrices sampled from \eqref{eq: matrix_distribution}, then by the law of large numbers
\begin{equation}
     \lim_{\substack{\beta \to 0 \\ N_s \to \infty}} \ \left \vert \frac{1}{N_s} \sum_{k=1}^{N_s} \tr(X_k^\T H X_k) - \frac{p}{n} \tr H  \right\vert = 0.
    \label{eq: frobenius_convergence}
\end{equation}
Next, let $\lambda_1 \le \lambda_2 \le \dots \le \lambda_n, \ \{\lambda_i \}_{i=1}^n = \mathrm{mspec}(H)$.
We now evaluate the expected total energy, i.e., the expected value under \eqref{eq: matrix_distribution} of the Rayleigh-Ritz trace $\langle E \rangle = \ \tr (X^\T H X)$. 
As $T_r \to 0$ ($\beta \to \infty$), the Boltzmann density concentrates entirely around the global minima of $E(X)$. By the Ky Fan minimum principle \citep{fan1949theorem}, this minimum is achieved when the columns of $X$ span the eigenspace of the $p$ lowest eigenvalues of $H$. That is, 
\end{minipage}
\hfill
\begin{minipage}[t]{0.41\linewidth}
\vspace{0pt}
\centering
\includegraphics[width=\linewidth]{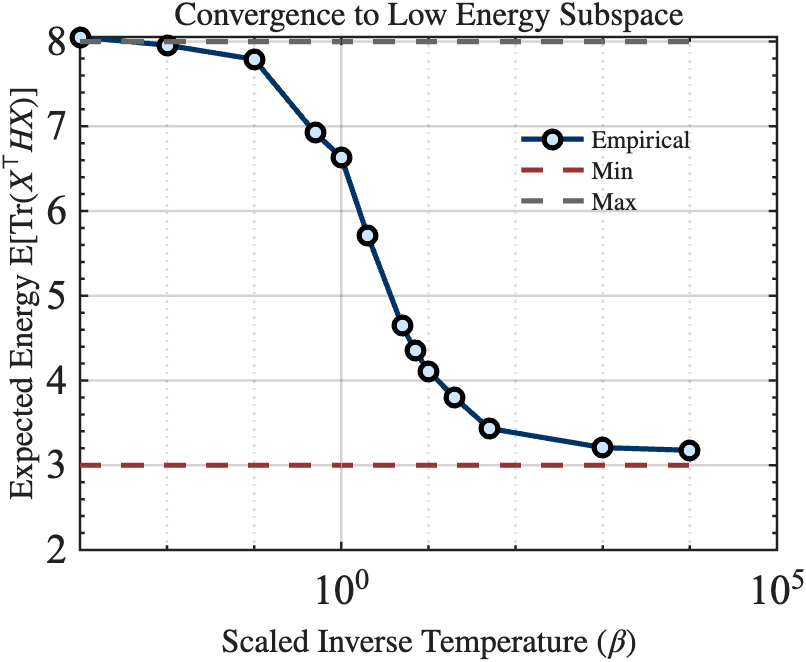}
\captionof{figure}{
The expectation was approximated by averaging the energy over $N_s = 5000$ orthogonal matrices sampled from \eqref{eq: matrix_distribution}.
The figure confirms the properties established in \eqref{eq: frobenius_convergence} and \eqref{eq: minimum_energy}. Gray and red dashed lines respectively denote $\sum_{i=1}^p \lambda_i$, and $\frac{p}{n} \tr H$, where $\mathrm{spec}(H) = \{1, 2, 5, 8\}$.}
\label{fig: stiefel_manifold}
\end{minipage}
\begin{equation}
    \lim_{\beta \to \infty} \mathbb{E} [ \tr (X^\T H X)] = \min_{X \in St(n,p)} \tr (X^\T H X) = \sum_{i=1}^p \lambda_i.
\label{eq: minimum_energy}
\end{equation}
Applying Algorithm~\ref{alg: extended_asbs_m} with $H$ defined in Appendix~\ref{app: app_stiefel_manifold} gives the results presented in Figure~\ref{fig: stiefel_manifold}.
As expected from standard statistical mechanics, we observe that $\frac{\diff }{\diff \beta} \langle E\rangle_\beta = - \mathrm{Var}_\beta (E) \leq 0$.

\subsection{High-Dimensional Redundant Inverse Kinematics for Closed-Loop Chains and Obstacle Avoidance}
To evaluate Extended R--ASBS on a high-dimensional space without geometric primitives, we consider the configuration space of closed-loop kinematic chains. Such systems frequently arise in parallel manipulators, multi-finger grasping, and cooperative multi-robot transport. 
We focus on generating diverse, valid configurations (static poses) avoiding a set of static obstacles.

Let $q \in \mathbb{T}^d$ denote the generalized joint coordinates\footnote{In this example, the ambient configuration space is the flat torus $\mathbb T^d$. Equivalently, computations are performed in an angular chart with all coordinates understood modulo \(2\pi\).} of the unconstrained articulated system.
The configuration space is defined by the nonlinear loop-closure constraints\footnote{We practically substitute $c_3(q)$ in \eqref{eq: robot_constraint} with $c_3(q) = \mathrm{atan}2 \left(\sin(\sum_i q_i - \vartheta_T), \cos(\sum_i q_i - \vartheta_T) \right)$.} $c(q)=0$, which enforce coincidence of the end-effector connection points.
In the scenario presented in Figure~\ref{fig: result_robotics}$(a)$, the robot is anchored at the base $(0,0)$, operates in a $2$D workspace, and
\begin{equation}
    c(q) = \begin{bmatrix}
        \sum_{i=1}^{10} L_i \cos \left(\sum_{j=1}^{i} q_j\right)- x_T\\
        \sum_{i=1}^{10} L_i \sin  \left(\sum_{j=1}^{i} q_j\right) - y_T\\
        \sum_{i=1}^{10} q_i - \vartheta_T
    \end{bmatrix},
    \label{eq: robot_constraint}
\end{equation}
where $(x_T, y_T)$ and $ \vartheta_T$ denote the target position of the end-effector block and its required orientation, respectively.
The energy landscape evaluates the distance from the links to the obstacles $\{ \mathcal{O}_k \}_{k=1}^{K}$
\begin{equation}
    E(q) = \sum_{k=1}^{K} \alpha \exp\left( - \frac{\mathrm{dist}(q, \mathcal{O}_k)^2}{2 \sigma_\mathrm{obs}^2} \right) + \lambda \| q - q_\mathrm{rest} \|^2, \notag
\end{equation}
where $\mathrm{dist}(\cdot)$ computes the minimum Euclidean distance between the robot's forward kinematic links and the obstacle centers, $\alpha$ scales the repulsion barrier, and $\lambda$ acts as a weak prior pulling the robot toward a comfortable resting posture $q_\mathrm{rest}$; see Figure~\ref{fig: result_robotics}$(b),(c)$.
\begin{figure}[!t]
\centering
\includegraphics[width=1.0\columnwidth]{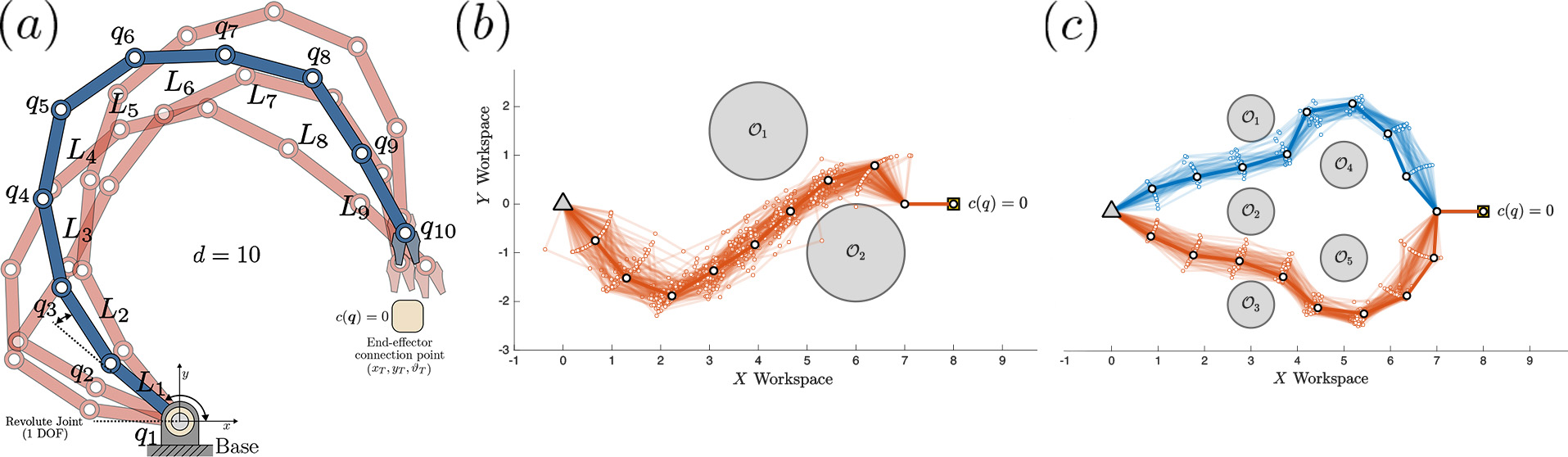}
\caption{$(a)$ Redundant planar closed-loop manipulator with $10$ revolute joints. $(b)$ The problem has obstacles $\mathcal{O}_1 := (4, 1.5)$, $\mathcal{O}_2 := (6, -1)$, with weak prior weight $\lambda = 0.05$. 
$(c)$ Same parameters, but with a different obstacle configuration. Because of the symmetry of the obstacles' displacement, the system exhibits two equiprobable modes (\textcolor{orange}{-} and \textcolor{blue}{-}). For illustration, random poses are highlighted.}
\label{fig: result_robotics}
\end{figure}

\subsection{The Wahba Problem with Outliers}
The Wahba problem \citep{wahba1965least, chin2019star}, also known as rotation search, is a fundamental problem in computer vision, robotics and aerospace engineering. It consists of finding the rotation between two coordinate frames given $N$ pairs of vector observations $a_i, b_i \in \R^3$, expressed in the two frames. 
Because many applications include severe percentages of outliers, we consider its robust formulation \citep{yang2019quaternion, yang2020teaser}, based on truncated least squares (TLS):
\begin{equation}
    \min_{\boldsymbol{R} \in \mathrm{SO}(3)} \sum_{i=1}^N \min \left( \frac{1}{\alpha_i^2} \|b_i - \boldsymbol{R} a_i\|^2, \bar{c}^2 \right),
    \label{eq:robust_wahba_problem}
\end{equation}
where $\alpha_i^2, \bar{c}^2$ are constants chosen by the designer \citep{yang2019quaternion}.
We adopt a quaternion formulation for \eqref{eq:robust_wahba_problem}. Unit quaternions are an alternative expression for $3$D rotations \citep{breckenridge1999quaternions} and are defined as a unit column vector $\boldsymbol{q} = [v^\T \ s]^\T$, where $v \in \R^3$ and $s$ are the vector and scalar parts, respectively. The set of quaternions is the $3$-Sphere manifold $\mathbb{S}^3 = \{\boldsymbol{q}\in \R^4 : \|\boldsymbol{q}\|=1 \}$. 
Define $\nu_{\mathbb{S}^3}(\boldsymbol{q}) \propto e^{-\beta J(\boldsymbol{q})}$, where $J$ is the objective of \eqref{eq:robust_wahba_problem} and $\beta > 0$ scale this energy-reward function.
Assume we have access to a large number of samples $N_s$, $\mathcal{B} := \{\boldsymbol{q}^{(i)}\}_{i=1}^{N_s} \sim \nu_{\mathbb{S}^3}, N_s \gg 1$. Then, $\boldsymbol{q}^\diamond := \argmin_{\boldsymbol{q} \in \mathcal{B}} J(\boldsymbol{q})$ would intuitively offer a close-to-optimal solution to \eqref{eq:robust_wahba_problem}.
Therefore, in this specific task, Algorithm~\ref{alg: extended_asbs_m} acts as a stochastic optimizer (Figure~\ref{fig: wahba_problem}).

Results are summarized in Table~\ref{tab: asbsm_wahba_results} and Appendix~\ref{app: app_robust_Wahba_problem}.
The table reports TLS objective values for $\boldsymbol{q}^\diamond$, and data are reported as means with standard deviations over $5$ independent repeats. It also reports rotation errors relative to the ground-truth rotation ($\vec{u} := [0.35; -0.75; 0.56], \vartheta := 72 \deg$), together with the number of true inliers recovered as active and true outliers clipped.
\begin{table}[H]
\centering
\caption{\underline{Extended} R--ASBS performance across outlier ratios
($N_s=10{,}000$, $\beta=1$).}
\label{tab: asbsm_wahba_results}
\vspace{-6pt}

\small
\newcommand{\sd}[1]{\graysd{#1}}

\renewcommand{\arraystretch}{1.05}
\setlength{\tabcolsep}{4.5pt}
\setlength{\aboverulesep}{1pt}
\setlength{\belowrulesep}{1pt}

\begin{adjustbox}{max width=\textwidth}
\begin{tabular}{c c c c c c c}
\toprule
Out. &
True TLS &
R--ASBS TLS &
TLS gap (\%)  &
R--ASBS err. (deg) &
Inliers active (\%) &
Out. clipped (\%)\\
\midrule
$25\%$
& $3439.58$ \sd{29.44}
& $3444.48$ \sd{29.05}
& $0.143$ \sd{0.146}
& $0.806$ \sd{0.246}
& $97.84$ \sd{0.59}
& $98.72$ \sd{0.66} \\

$50\%$
& $4832.57$ \sd{31.46}
& $4833.15$ \sd{31.23}
& $0.012$ \sd{0.017}
& $0.733$ \sd{0.258}
& $97.88$ \sd{1.01}
& $98.44$ \sd{0.57} \\

$75\%$
& $6285.06$ \sd{47.37}
& $6283.50$ \sd{47.87}
& $-0.025$ \sd{0.025}
& $1.119$ \sd{0.316}
& $98.16$ \sd{0.78}
& $97.95$ \sd{0.50} \\

$85\%$
& $6874.01$ \sd{22.90}
& $6873.50$ \sd{23.16}
& $-0.007$ \sd{0.012}
& $0.892$ \sd{0.688}
& $98.27$ \sd{0.76}
& $98.12$ \sd{0.35} \\

$90\%$
& $7164.73$ \sd{32.43}
& $7161.57$ \sd{33.36}
& $-0.044$ \sd{0.052}
& $1.253$ \sd{0.429}
& $97.20$ \sd{2.05}
& $98.13$ \sd{0.30} \\

$95\%$
& $7446.11$ \sd{12.51}
& $7441.82$ \sd{10.99}
& $-0.058$ \sd{0.029}
& $2.245$ \sd{1.111}
& $96.00$ \sd{2.83}
& $97.73$ \sd{0.40} \\
\bottomrule
\end{tabular}
\end{adjustbox}
\end{table}

\section{Conclusion and Future Work}
We proposed a theoretically supported enlargement of the ASBS sampler to settings where the unnormalized probability distribution is intrinsically defined on non-Euclidean spaces.
The continuous-time formulation is justified through a manifold SOC--SB equivalence, while the practical R--ASBS algorithm is derived as a geometry-aware approximation of the resulting matching conditions. Extensive numerical results demonstrate the efficacy of the method.
An interesting direction, that by now we leave as future work, is expanding the theoretical certificates on convergence.

\section*{Acknowledgments}
The authors would like to thank Paolo Giaretta and Tong Wang for their helpful discussions and comments. They also express their gratitude to Luigi Tesio for his assistance in designing Figure~\ref{fig: result_robotics}.

\bibliography{refs}

@article{fan1949theorem,
  author  = {Fan, Ky},
  title   = {On a Theorem of {Weyl} Concerning Eigenvalues of Linear Transformations {I}},
  journal = {Proceedings of the National Academy of Sciences},
  volume  = {35},
  number  = {11},
  pages   = {652--655},
  year    = {1949},
  doi     = {10.1073/pnas.35.11.652}
}

@book{absil2008optimization,
  author    = {Absil, P.-A. and Mahony, Robert and Sepulchre, Rodolphe},
  title     = {Optimization Algorithms on Matrix Manifolds},
  publisher = {Princeton University Press},
  year      = {2008}
}

@article{grigoryan1997gaussian,
  author  = {Grigor'yan, Alexander},
  title   = {Gaussian Upper Bounds for the Heat Kernel on Arbitrary Manifolds},
  journal = {Journal of Differential Geometry},
  volume  = {45},
  number  = {1},
  pages   = {33--52},
  year    = {1997},
  doi     = {10.4310/jdg/1214459753}
}

@article{jamison1975markov,
  author  = {Jamison, Benton},
  title   = {The {Markov} Processes of {Schr{\"o}dinger}},
  journal = {Zeitschrift f{\"u}r Wahrscheinlichkeitstheorie und Verwandte Gebiete},
  volume  = {32},
  number  = {4},
  pages   = {323--331},
  year    = {1975},
  doi     = {10.1007/BF00535844}
}

@book{ikeda1989stochastic,
  author    = {Ikeda, Nobuyuki and Watanabe, Shinzo},
  title     = {Stochastic Differential Equations and Diffusion Processes},
  edition   = {2},
  series    = {North-Holland Mathematical Library},
  volume    = {24},
  publisher = {North-Holland},
  address   = {Amsterdam},
  year      = {1989}
}

@book{hsu2002stochastic,
  author    = {Hsu, Elton P.},
  title     = {Stochastic Analysis on Manifolds},
  series    = {Graduate Studies in Mathematics},
  volume    = {38},
  publisher = {American Mathematical Society},
  address   = {Providence, RI},
  year      = {2002},
  doi       = {10.1090/gsm/038}
}

@article{thalmaier1998remarks,
  author  = {Thalmaier, Anton},
  title   = {Some Remarks on the Heat Flow for Functions and Forms},
  journal = {Electronic Communications in Probability},
  volume  = {3},
  pages   = {43--49},
  year    = {1998},
  doi     = {10.1214/ECP.v3-992}
}

@article{coulibaly2011brownian,
  author  = {Coulibaly-Pasquier, Kol{\'e}h{\`e} A.},
  title   = {{Brownian} Motion with Respect to Time-Changing {Riemannian} Metrics, Applications to {Ricci} Flow},
  journal = {Annales de l'Institut Henri Poincar{\'e}, Probabilit{\'e}s et Statistiques},
  volume  = {47},
  number  = {2},
  pages   = {515--538},
  year    = {2011},
  doi     = {10.1214/10-AIHP364}
}

@book{lee2018introduction,
  author    = {Lee, John M.},
  title     = {Introduction to {Riemannian} Manifolds},
  series    = {Graduate Texts in Mathematics},
  volume    = {176},
  edition   = {2},
  publisher = {Springer},
  address   = {Cham},
  year      = {2018},
  doi       = {10.1007/978-3-319-91755-9}
}

@article{lang2007bayesian,
  author  = {Lang, Lixin and Chen, Wen-shiang and Bakshi, Bhavik R. and Goel, Prem K. and Ungarala, Sridhar},
  title   = {{Bayesian} Estimation via Sequential {Monte Carlo} Sampling---Constrained Dynamic Systems},
  journal = {Automatica},
  volume  = {43},
  number  = {9},
  pages   = {1615--1622},
  year    = {2007},
  doi     = {10.1016/j.automatica.2007.02.012}
}

@article{li2015efficient,
  author  = {Li, Yifang and Ghosh, Sujit K.},
  title   = {Efficient Sampling Methods for Truncated Multivariate Normal and {Student-t} Distributions Subject to Linear Inequality Constraints},
  journal = {Journal of Statistical Theory and Practice},
  volume  = {9},
  number  = {4},
  pages   = {712--732},
  year    = {2015},
  doi     = {10.1080/15598608.2014.996690}
}

@article{girolami2011riemann,
  author  = {Girolami, Mark and Calderhead, Ben},
  title   = {Riemann Manifold {Langevin} and {Hamiltonian} {Monte Carlo} Methods},
  journal = {Journal of the Royal Statistical Society: Series B (Statistical Methodology)},
  volume  = {73},
  number  = {2},
  pages   = {123--214},
  year    = {2011},
  doi     = {10.1111/j.1467-9868.2010.00765.x}
}

@inproceedings{brubaker2012family,
  author    = {Brubaker, Marcus and Salzmann, Mathieu and Urtasun, Raquel},
  title     = {A Family of {MCMC} Methods on Implicitly Defined Manifolds},
  booktitle = {Proceedings of the Fifteenth International Conference on Artificial Intelligence and Statistics},
  series    = {Proceedings of Machine Learning Research},
  volume    = {22},
  pages     = {161--172},
  publisher = {PMLR},
  year      = {2012},
  url       = {https://proceedings.mlr.press/v22/brubaker12.html}
}

@inproceedings{debortoli2022riemannian,
  author    = {De Bortoli, Valentin and Mathieu, Emile and Hutchinson, Michael and Thornton, James and Teh, Yee Whye and Doucet, Arnaud},
  title     = {{Riemannian} Score-Based Generative Modelling},
  booktitle = {Advances in Neural Information Processing Systems},
  volume    = {35},
  pages     = {2406--2422},
  year      = {2022}
}

@inproceedings{lou2023reflected,
  author    = {Lou, Aaron and Ermon, Stefano},
  title     = {Reflected Diffusion Models},
  booktitle = {Proceedings of the 40th International Conference on Machine Learning},
  series    = {Proceedings of Machine Learning Research},
  volume    = {202},
  pages     = {22675--22701},
  publisher = {PMLR},
  year      = {2023},
  url       = {https://proceedings.mlr.press/v202/lou23a.html}
}

@techreport{breckenridge1999quaternions,
  author      = {Breckenridge, William G.},
  title       = {Quaternions Proposed Standard Conventions},
  institution = {Jet Propulsion Laboratory},
  type        = {Interoffice Memorandum},
  number      = {IOM 343-79-1199},
  year        = {1999}
}

@inproceedings{yang2019quaternion,
  author    = {Yang, Heng and Carlone, Luca},
  title     = {A Quaternion-Based Certifiably Optimal Solution to the {Wahba} Problem With Outliers},
  booktitle = {Proceedings of the IEEE/CVF International Conference on Computer Vision (ICCV)},
  pages     = {1665--1674},
  year      = {2019},
  doi       = {10.1109/ICCV.2019.00175}
}

@article{yang2020teaser,
  author  = {Yang, Heng and Shi, Jingnan and Carlone, Luca},
  title   = {{TEASER}: Fast and Certifiable Point Cloud Registration},
  journal = {IEEE Transactions on Robotics},
  volume  = {37},
  number  = {2},
  pages   = {314--333},
  year    = {2021},
  doi     = {10.1109/TRO.2020.3033695}
}

@inproceedings{chin2019star,
  author    = {Chin, Tat-Jun and Bagchi, Samya and Eriksson, Anders and van Schaik, Andre},
  title     = {Star Tracking Using an Event Camera},
  booktitle = {Proceedings of the IEEE/CVF Conference on Computer Vision and Pattern Recognition Workshops (CVPRW)},
  pages     = {1646--1655},
  year      = {2019},
  doi       = {10.1109/CVPRW.2019.00208}
}

@article{wahba1965least,
  author  = {Wahba, Grace},
  title   = {A Least Squares Estimate of Satellite Attitude},
  journal = {SIAM Review},
  volume  = {7},
  number  = {3},
  pages   = {409},
  year    = {1965},
  doi     = {10.1137/1007077}
}

@book{grigoryan2009heat,
  author    = {Grigor'yan, Alexander},
  title     = {Heat Kernel and Analysis on Manifolds},
  series    = {AMS/IP Studies in Advanced Mathematics},
  volume    = {47},
  publisher = {American Mathematical Society and International Press},
  year      = {2009}
}

@misc{USGSComCat,
  author       = {{U.S. Geological Survey, Earthquake Hazards Program}},
  title        = {{Advanced National Seismic System (ANSS)} Comprehensive Catalog of Earthquake Events and Products},
  howpublished = {U.S. Geological Survey},
  year         = {2017},
  url          = {https://doi.org/10.5066/F7MS3QZH}
}

@misc{PierreAugerOpenData,
  author       = {{Pierre Auger Collaboration}},
  title        = {{Pierre Auger Observatory Open Data}},
  howpublished = {Zenodo},
  year         = {2024},
  url          = {https://doi.org/10.5281/zenodo.10488964}
}

@inproceedings{TzenRaginsky2019,
  author    = {Tzen, Belinda and Raginsky, Maxim},
  title     = {Theoretical Guarantees for Sampling and Inference in Generative Models with Latent Diffusions},
  booktitle = {Proceedings of the Thirty-Second Conference on Learning Theory},
  series    = {Proceedings of Machine Learning Research},
  volume    = {99},
  pages     = {3084--3114},
  publisher = {PMLR},
  year      = {2019},
  url       = {https://proceedings.mlr.press/v99/tzen19a.html}
}

@article{schrodinger1931umkehrung,
  author  = {Schr{\"o}dinger, Erwin},
  title   = {{\"U}ber die {Umkehrung} der {Naturgesetze}},
  journal = {Sitzungsberichte der Preussischen Akademie der Wissenschaften, Physikalisch-Mathematische Klasse},
  pages   = {144--153},
  year    = {1931}
}

@inproceedings{Guo2026DiscreteASBS,
  author    = {Guo, Wei and Zhu, Yuchen and Du, Xiaochen and Nam, Juno and Chen, Yongxin and G{\'o}mez-Bombarelli, Rafael and Liu, Guan-Horng and Tao, Molei and Choi, Jaemoo},
  title     = {Discrete Adjoint {Schr{\"o}dinger} Bridge Sampler},
  booktitle = {Forty-third International Conference on Machine Learning (ICML)},
  year      = {2026}
}

@article{cui2024lowrank,
  author  = {Cui, Tiangang and Gorodetsky, Alex A.},
  title   = {Low-rank {Bayesian} Matrix Completion via Geodesic {Hamiltonian Monte Carlo} on {Stiefel} Manifolds},
  journal = {Foundations of Data Science},
  year    = {2026},
  note    = {Early Access},
  doi     = {10.3934/fods.2026015}
}

@article{pourzanjani2021bayesian,
  author  = {Pourzanjani, Arya A. and Jiang, Richard M. and Mitchell, Brian and Atzberger, Paul J. and Petzold, Linda R.},
  title   = {{Bayesian} Inference over the {Stiefel} Manifold via the {Givens} Representation},
  journal = {Bayesian Analysis},
  volume  = {16},
  number  = {2},
  pages   = {639--666},
  year    = {2021},
  doi     = {10.1214/20-BA1202}
}

@article{jauch2021bayesian,
  author  = {Jauch, Michael and Hoff, Peter D. and Dunson, David B.},
  title   = {{Monte Carlo} Simulation on the {Stiefel} Manifold via Polar Expansion},
  journal = {Journal of Computational and Graphical Statistics},
  volume  = {30},
  number  = {3},
  pages   = {622--631},
  year    = {2021},
  doi     = {10.1080/10618600.2020.1859382}
}

@inproceedings{havens2025adjoint,
  author    = {Havens, Aaron J. and Miller, Benjamin Kurt and Yan, Bing and Domingo-Enrich, Carles and Sriram, Anuroop and Levine, Daniel S. and Wood, Brandon M. and Hu, Bin and Amos, Brandon and Karrer, Brian and Fu, Xiang and Liu, Guan-Horng and Chen, Ricky T. Q.},
  title     = {Adjoint Sampling: Highly Scalable Diffusion Samplers via Adjoint Matching},
  booktitle = {Proceedings of the 42nd International Conference on Machine Learning},
  series    = {Proceedings of Machine Learning Research},
  volume    = {267},
  pages     = {22204--22237},
  publisher = {PMLR},
  year      = {2025},
  url       = {https://proceedings.mlr.press/v267/havens25a.html}
}

@inproceedings{domingoenrich2025adjoint,
  author    = {Domingo-Enrich, Carles and Drozdzal, Michal and Karrer, Brian and Chen, Ricky T. Q.},
  title     = {Adjoint Matching: Fine-tuning Flow and Diffusion Generative Models with Memoryless Stochastic Optimal Control},
  booktitle = {International Conference on Learning Representations (ICLR)},
  year      = {2025},
  url       = {https://openreview.net/forum?id=xQBRrtQM8u}
}

@inproceedings{guan2025riemannianproximal,
  author    = {Guan, Yunrui and Balasubramanian, Krishnakumar and Ma, Shiqian},
  title     = {{Riemannian} Proximal Sampler for High-Accuracy Sampling on Manifolds},
  booktitle = {Advances in Neural Information Processing Systems},
  volume    = {38},
  year      = {2025}
}

@inproceedings{havens2026flowsampling,
  author    = {Havens, Aaron J. and Karrer, Brian and Shaul, Neta},
  title     = {Flow Sampling: Learning to Sample from Unnormalized Densities via Denoising Conditional Processes},
  booktitle = {Forty-third International Conference on Machine Learning (ICML)},
  year      = {2026}
}

@inproceedings{noble2023unbiased,
  author    = {Noble, Maxence and De Bortoli, Valentin and Durmus, Alain},
  title     = {Unbiased Constrained Sampling with Self-Concordant Barrier {Hamiltonian Monte Carlo}},
  booktitle = {Advances in Neural Information Processing Systems},
  volume    = {36},
  year      = {2023}
}

@inproceedings{cheng2022efficient,
  author    = {Cheng, Xiang and Zhang, Jingzhao and Sra, Suvrit},
  title     = {Efficient Sampling on {Riemannian} Manifolds via {Langevin} {MCMC}},
  booktitle = {Advances in Neural Information Processing Systems},
  volume    = {35},
  year      = {2022}
}

@misc{thornton2022riemannian,
  author       = {Thornton, James and Hutchinson, Michael and Mathieu, Emile and De Bortoli, Valentin and Teh, Yee Whye and Doucet, Arnaud},
  title        = {{Riemannian} Diffusion {Schr{\"o}dinger} Bridge},
  howpublished = {arXiv:2207.03024},
  year         = {2022},
  url          = {https://arxiv.org/abs/2207.03024}
}

@inproceedings{deng2024reflected,
  author    = {Deng, Wei and Chen, Yu and Yang, Nicole Tianjiao and Du, Hengrong and Feng, Qi and Chen, Ricky Tian Qi},
  title     = {Reflected {Schr{\"o}dinger} Bridge for Constrained Generative Modeling},
  booktitle = {Proceedings of the Fortieth Conference on Uncertainty in Artificial Intelligence},
  series    = {Proceedings of Machine Learning Research},
  volume    = {244},
  pages     = {1055--1082},
  publisher = {PMLR},
  year      = {2024},
  url       = {https://proceedings.mlr.press/v244/deng24b.html}
}

@inproceedings{caluya2020reflected,
  author    = {Caluya, Kenneth F. and Halder, Abhishek},
  title     = {Reflected {Schr{\"o}dinger} Bridge: Density Control with Path Constraints},
  booktitle = {2021 American Control Conference (ACC)},
  pages     = {1137--1142},
  year      = {2021},
  doi       = {10.23919/ACC50511.2021.9482813}
}

@inproceedings{liu2023mirror,
  author    = {Liu, Guan-Horng and Chen, Tianrong and Theodorou, Evangelos and Tao, Molei},
  title     = {Mirror Diffusion Models for Constrained and Watermarked Generation},
  booktitle = {Advances in Neural Information Processing Systems},
  volume    = {36},
  year      = {2023}
}

@inproceedings{fishman2023metropolis,
  author    = {Fishman, Nic and Klarner, Leo and Mathieu, Emile and Hutchinson, Michael and De Bortoli, Valentin},
  title     = {Metropolis Sampling for Constrained Diffusion Models},
  booktitle = {Advances in Neural Information Processing Systems},
  volume    = {36},
  year      = {2023}
}

@article{fishman2023diffusion,
  author  = {Fishman, Nic and Klarner, Leo and De Bortoli, Valentin and Mathieu, Emile and Hutchinson, Michael John},
  title   = {Diffusion Models for Constrained Domains},
  journal = {Transactions on Machine Learning Research},
  year    = {2023},
  url     = {https://openreview.net/forum?id=xuWTFQ4VGO}
}

@inproceedings{christopher2024constrained,
  author    = {Christopher, Jacob K. and Baek, Stephen and Fioretto, Ferdinando},
  title     = {Constrained Synthesis with Projected Diffusion Models},
  booktitle = {Advances in Neural Information Processing Systems},
  volume    = {37},
  year      = {2024}
}

@inproceedings{chamon2024constrained,
  author    = {Chamon, Luiz F. O. and Karimi, Mohammad Reza and Korba, Anna},
  title     = {Constrained Sampling with Primal-Dual {Langevin} {Monte Carlo}},
  booktitle = {Advances in Neural Information Processing Systems},
  volume    = {37},
  year      = {2024}
}

@inproceedings{blanke2026strictly,
  author    = {Blanke, Matthieu and Qu, Yongquan and Shamekh, Sara and Gentine, Pierre},
  title     = {Strictly Constrained Generative Modeling via Split Augmented {Langevin} Sampling},
  booktitle = {International Conference on Learning Representations (ICLR)},
  year      = {2026}
}

@inproceedings{liu2025adjoint,
  author    = {Liu, Guan-Horng and Choi, Jaemoo and Chen, Yongxin and Miller, Benjamin Kurt and Chen, Ricky T. Q.},
  title     = {Adjoint {Schr{\"o}dinger} Bridge Sampler},
  booktitle = {Advances in Neural Information Processing Systems},
  volume    = {38},
  year      = {2025}
}

@article{delmoral2006sequential,
  author  = {Del Moral, Pierre and Doucet, Arnaud and Jasra, Ajay},
  title   = {Sequential {Monte Carlo} Samplers},
  journal = {Journal of the Royal Statistical Society: Series B (Statistical Methodology)},
  volume  = {68},
  number  = {3},
  pages   = {411--436},
  year    = {2006},
  doi     = {10.1111/j.1467-9868.2006.00553.x}
}

@article{kullback1951information,
  author  = {Kullback, Solomon and Leibler, Richard A.},
  title   = {On Information and Sufficiency},
  journal = {The Annals of Mathematical Statistics},
  volume  = {22},
  number  = {1},
  pages   = {79--86},
  year    = {1951},
  doi     = {10.1214/aoms/1177729694}
}

@inproceedings{zhang2022path,
  author    = {Zhang, Qinsheng and Chen, Yongxin},
  title     = {Path Integral Sampler: A Stochastic Control Approach for Sampling},
  booktitle = {International Conference on Learning Representations (ICLR)},
  year      = {2022},
  url       = {https://openreview.net/forum?id=_uCb2ynRu7Y}
}

@article{neal2001annealed,
  author  = {Neal, Radford M.},
  title   = {Annealed Importance Sampling},
  journal = {Statistics and Computing},
  volume  = {11},
  number  = {2},
  pages   = {125--139},
  year    = {2001},
  doi     = {10.1023/A:1008923215028}
}

@article{metropolis1953equation,
  author  = {Metropolis, Nicholas and Rosenbluth, Arianna W. and Rosenbluth, Marshall N. and Teller, Augusta H. and Teller, Edward},
  title   = {Equation of State Calculations by Fast Computing Machines},
  journal = {The Journal of Chemical Physics},
  volume  = {21},
  number  = {6},
  pages   = {1087--1092},
  year    = {1953},
  doi     = {10.1063/1.1699114}
}

@article{chen2016relation,
  author  = {Chen, Yongxin and Georgiou, Tryphon T. and Pavon, Michele},
  title   = {On the Relation Between Optimal Transport and {Schr{\"o}dinger} Bridges: A Stochastic Control Viewpoint},
  journal = {Journal of Optimization Theory and Applications},
  volume  = {169},
  number  = {2},
  pages   = {671--691},
  year    = {2016},
  doi     = {10.1007/s10957-015-0803-z}
}

@book{sarkka2019applied,
  author    = {S{\"a}rkk{\"a}, Simo and Solin, Arno},
  title     = {Applied Stochastic Differential Equations},
  series    = {Institute of Mathematical Statistics Textbooks},
  volume    = {10},
  publisher = {Cambridge University Press},
  year      = {2019},
  doi       = {10.1017/9781108186735}
}

@incollection{follmer1988random,
  author    = {F{\"o}llmer, Hans},
  title     = {Random Fields and Diffusion Processes},
  booktitle = {{\'E}cole d'{\'E}t{\'e} de Probabilit{\'e}s de Saint-Flour XV--XVII, 1985--87},
  series    = {Lecture Notes in Mathematics},
  volume    = {1362},
  pages     = {101--203},
  publisher = {Springer},
  address   = {Berlin},
  year      = {1988},
  doi       = {10.1007/BFb0086180}
}

@article{leonard2014survey,
  author  = {L{\'e}onard, Christian},
  title   = {A Survey of the {Schr{\"o}dinger} Problem and Some of Its Connections with Optimal Transport},
  journal = {Discrete and Continuous Dynamical Systems},
  volume  = {34},
  number  = {4},
  pages   = {1533--1574},
  year    = {2014},
  doi     = {10.3934/dcds.2014.34.1533}
}
\bibliographystyle{arxiv_style}

\newpage
\appendix
\addtocontents{toc}{\protect\setcounter{tocdepth}{2}}
\tableofcontents

\section{Related work}
\label{app: related_work}
It is noteworthy that in constrained sampling, we distinguish between two types of constraints: support constraints, which correspond to our case of study \eqref{eq: M_constraint}, and statistical constraints \citep{chamon2024constrained}.
A classical family of methods for the former relies on rejection sampling (see, e.g., \citep{lang2007bayesian, li2015efficient}). 
While such methods preserve feasibility, they become inefficient, in terms of number of samples generated per iteration, in high-dimensional spaces or under intricate nonlinear constraints. 
Alternative methods have been developed for target distributions supported on manifolds \citep{girolami2011riemann, brubaker2012family}, including split augmented Langevin sampling \citep{blanke2026strictly}, barrier Hamiltonian Monte Carlo \citep{noble2023unbiased}, Riemannian Langevin MCMC \citep{cheng2022efficient}, and Riemannian proximal sampling \citep{guan2025riemannianproximal}. 
These methods provide standardized methodologies for respecting geometric constraints, either by exploiting intrinsic differential geometry, augmenting the constrained dynamics, or using barrier/proximal constructions to avoid infeasible regions. Their main advantage is that they inherit the asymptotic correctness guarantees of classical MCMC-type samplers under appropriate regularity conditions.
However, they are typically iterative Markov-chain methods: they may require long mixing times in multimodal landscapes, constraint-solving steps, and careful tuning of step sizes or auxiliary dynamics. These costs become especially pronounced for high-dimensional, nonlinear, or sharply constrained manifolds (a better intuition of this behavior is given in Section~\ref{subsec: sphere_climate}).

In contrast, the present work follows a \emph{controlled-diffusion} perspective. Rather than constructing a Markov chain whose stationary distribution is $\nu_\mathcal{M}$, we learn a time-dependent \emph{tangent control field} that transports a tractable initial source distribution on $\mathcal{M}$ toward the target Boltzmann distribution at terminal time. 

With a similar philosophy, flow sampling \citep{havens2026flowsampling} provides a grounded framework for learning a diffusion process whose probability path interpolates between a source and a target density through a denoising conditional process and flow-matching-style regression. 
Despite its flexibility and the absence of a corrector network, flow sampling, even if admits a
Riemannian formulation, has closed-form conditional drifts currently restricted to constant-curvature manifolds such as hyperspheres and
hyperbolic spaces.
Instead, our paper formulates unnormalized sampling on embedded equality-constraint manifolds as a manifold SOC--SB problem, allowing more general constrained spaces and retaining an ASBS-style corrector for non-memoryless source distributions.

We also mention, as a complementary line of work, constrained generative modeling with diffusion processes. 
Starting from the availability of data, rather than access to the energy function and its gradient as assumed in the present paper, diffusion models are capable of synthesizing data living on Riemannian manifolds \citep{debortoli2022riemannian}. They may rely on barrier functions, mirror maps \citep{liu2023mirror}, reflected Brownian motion \citep{lou2023reflected, fishman2023diffusion}, or broader constraint constructions \citep{fishman2023metropolis, christopher2024constrained}.
Schr\"odinger bridge formulations have also been extended to settings with constraints \citep{caluya2020reflected, deng2024reflected} or non-Euclidean geometry \citep{thornton2022riemannian}.

\section{Adjoint Matching and Reciprocal Adjoint Matching}
\label{app: adjoint}
This appendix reviews the general continuous-adjoint formulation underlying Adjoint Matching, Adjoint Sampling, and Reciprocal Adjoint Matching.
\subsection{Adjoint Matching}
\label{app: adjoint_matching} 
The SOC problem considered in adjoint SB sampling takes the form:
\begin{equation}
    \min_{u \in \mathcal{U}} \ \mathbb{E}_{\boldsymbol{X} \sim p^u} \left[ \int_0^1 \frac12 \|u_t(X_t)\|^2 \diff t + g(X_1) \right],
    \qquad \text{s.t.} \quad \eqref{eq: flat_controlled_sde},
\label{eq: flat_SOC}
\end{equation}
where $g(x):\R^d \to \R$ is a terminal cost.

Scalable computational methods for solving \eqref{eq: flat_SOC} have been challenging, as naively back-propagating through controlled stochastic trajectories causes prohibitively high computational cost.
Instead, \citet{domingoenrich2025adjoint} presents \emph{Adjoint Matching}\footnote{Note \citet{domingoenrich2025adjoint} considers a more general setting including in the cost functional the running cost $r_t$, which is left for completeness but in our setting we may fix $r_t \equiv 0$. Moreover, no scalar assumption is imposed on $\sigma_t$.}, a regression-based method for solving entropy-regularized stochastic optimal control problems arising in reward fine-tuning of flow and diffusion generative models.
AM essentially combines the continuous adjoint method with least-squares control matching, turning \eqref{eq: flat_SOC} into a supervised regression problem over controlled trajectories.
To this end, we first define the adjoint state along a trajectory $\boldsymbol{X} \sim p^u$ as:
\begin{equation}
    a_t(\boldsymbol{X}, u) := \nabla_{X_t} \bigg[ \int_t^1  \bigg(\frac12 \|u_s(X_s)\|^2 + r_s(X_s) \bigg) \diff s + g(X_1) \bigg], \notag
\end{equation}
where $\boldsymbol{X}$ solves \eqref{eq: flat_controlled_sde}. This implies that 
\begin{equation}
    \mathbb{E}_{\boldsymbol{X} \sim p^u} [a_t(\boldsymbol{X}, u) \vert X_t = x] = \nabla_x J^u(x,t), \notag
\end{equation}
where $J^u$ is the cost functional at $(x,t)$ under the policy $u$.
It can be shown that $a_t \equiv a(t; \boldsymbol{X}_{[t,1]})$ satisfies the backward ODE equation:
\begin{equation}
    \frac{\diff}{\diff t} a_t(\boldsymbol{X},u) = -\nabla_{X_t} \big[(f_t + \sigma_t u_t)(X_t)\big]^T a_t(\boldsymbol{X},u) - \nabla_{X_t} \Big(r_t(X_t) + \frac12 \|u_t(X_t)\|^2 \Big),
\label{eq: adj_diff_eq}
\end{equation}
with boundary condition $a_1(\boldsymbol{X},u) = \nabla_{X_1} g(X_1)$.
The corresponding basic AM objective is
\begin{subequations}
\begin{align}
    &L_{\text{bAM}}(u; \boldsymbol{X}) := \frac{1}{2} \int_0^1 \| u_t(X_t) + \sigma_t a_t(\boldsymbol{X}, \bar{u})\|^2 \diff t, \\
    &\text{s.t.} \quad \boldsymbol{X} \sim p^{\bar{u}}, \quad \bar{u} = \texttt{stopgrad}(u),
\end{align}   
\label{eq: basic_AM}
\end{subequations}
where $\bar{u} = \texttt{stopgrad}(u)$ means that the gradients of $\bar{u}$ with respect to the parameters $\theta$ of the control $u$ are artificially set to zero. 
Via this operator, although $\boldsymbol{X} \sim p^{\bar{u}}$ is sampled according to the control process, AM does not differentiate through the sampling procedure.
Importantly, the only critical point of $\mathbb{E}[L_{\text{bAM}}]$ is the optimal control $u^\star$.
We remark that \eqref{eq: basic_AM} allows us to construct a matching objective without any importance weighting, circumventing the issue of high variance importance weights while retaining the interpretation of matching a vector field.

For practical implementation, to obviate the Jacobian of the control $\nabla_x u$ in \eqref{eq: adj_diff_eq}, AM replaces the full adjoint by the lean adjoint $\tilde{a}$.
Conceptually, AM directly tries to find the fixed point $u_t(x) = -\sigma_t \nabla_x J^u(x,t)$ by replacing $\nabla_x J^u(x,t)$ with a stochastic estimator, which is the lean adjoint.
Then, the unique solution to the fixed point relation
\begin{equation*}
    u_t (x) = - \sigma_t \mathbb{E}_{\boldsymbol{X} \sim p^u}[\tilde{a}_t(\boldsymbol{X}) \vert X_t = x]
\end{equation*}
is $u^\star$ (we defer the formal proofs to \citet[Appendix E]{domingoenrich2025adjoint}).
Thus, this motivates the following objective as a means to solve this fixed point problem:
\begin{subequations}
\begin{align}
    &L_{\text{AM}}(u) := \mathbb{E}_{p^{\bar{u}}} \left[ \frac{1}{2} \int_0^1 \| u_t(X_t) + \sigma_t \tilde{a}_t(\boldsymbol{X})\|^2 \diff t \right], \notag \\
    &\text{s.t.} \quad \boldsymbol{X} \sim p^{\bar{u}}, \quad \bar{u} = \texttt{stopgrad}(u), \\
    &\hphantom{\text{s.t.}} \quad \frac{\diff}{\diff t} \tilde{a}_t(\boldsymbol{X}) = - [\nabla_x f_t(X_t)]^\T \tilde{a}_t(\boldsymbol{X}) - \nabla_x r_t(X_t), \label{eq: ode_a_lean}\\
    &\hphantom{\text{s.t.}} \quad \tilde{a}_1(\boldsymbol{X}) = \nabla_{X_1} g(X_1). \label{eq: a_lean_terminal_cond} 
\end{align}   
\end{subequations}
Note in this variant, the backward ODE \eqref{eq: ode_a_lean} is driven only by the base system dynamics; thus, the neural network (NN) $u_t^\theta(x)$ has been completely excised from the backward integration.

\subsection{Adjoint Sampling and Reciprocal Adjoint Matching}
\label{app: adjoint_sampling}
\emph{Adjoint sampling} (AS) \citep{havens2025adjoint} specializes AM to the problem of sampling from an unnormalized Gibbs--Boltzmann distribution.
The method formulates sampling as a SOC problem over a controlled diffusion initialized from a degenerate source.

Let $p^\mathrm{base} (X_0, X_1)$ denote the joint endpoint density under the base process, and let $p_{1|0}^\mathrm{base} (y|x)$ be the corresponding transition density.
Rewriting \eqref{eq: flat_SOC} as $D_\mathrm{KL} (p^u \vert\vert p^{\mathrm{base}}) + \mathbb{E}_{X_1 \sim p_1^u}[g(X_1)]$ and then computing the analytical solution $p^\star(X_1|X_0) \propto p^\mathrm{base} (X_1|X_0) e^{-g(X_1)}$ and normalization $\int p^\mathrm{base} (X_1 | X_0)$ $e^{-g(X_1)} \diff X_1 = e^{-V_0(X_0)}$, leads to the closed-form optimal distribution \citep[Thm. 1]{TzenRaginsky2019}
\begin{equation}
    p^\star(X_0, X_1) = p^\mathrm{base} (X_0, X_1) e^{-g(X_1) + \myellow{V_0 (X_0)}},
    \label{eq: p_star_X0_X1}
\end{equation}
where $V_0(x) = - \log \int p_{1|0}^\mathrm{base} (y|x) e^{-g(y)} \diff y$ is the initial value function.

Indeed, marginalizing \eqref{eq: p_star_X0_X1} over $X_0$ gives the terminal distribution $p^\star(X_1)$, which generally depends not only on $g(X_1)$ but also on the nontrivial coupling between $X_0$ and $X_1$ under the base process.
Consequently, to ensure that the terminal marginal $p^\star(X_1)$ coincides with the target density $\nu(X_1)$, we must eliminate the extra factor induced by $\myellow{V_0(X_0)}$.

A common approach to discard the aforementioned initial value function bias is to impose a memoryless condition on the base process. Formally, it necessitates that the initial and terminal variables be statistically independent under the base process path measure
\begin{equation}
    p^\mathrm{base} (X_0, X_1) = p^\mathrm{base}(X_0) p^\mathrm{base}(X_1).
    \label{eq: memoryless_condition}
\end{equation}
Under \eqref{eq: memoryless_condition}, the initial value function becomes independent of $X_0$, simplifying the optimal distribution at terminal time $t=1$ to
\begin{align}
    p_1^\star(X_1) &= \int p^\mathrm{base} (X_0) p^\mathrm{base}(X_1) e^{-g(X_1) + \myellow{V_0(X_0)}} \diff X_0 \notag \\
    &\propto p^\mathrm{base} (X_1) e^{-g(X_1)} = \nu(X_1), \label{eq: g_terminal_cost}
\end{align}
where \eqref{eq: g_terminal_cost} follows from setting $g(x) := \log \frac{p_1^\mathrm{base} (x)}{\nu(x)}.$
Therefore, under \eqref{eq: memoryless_condition}, the desired target $\nu$ can be imposed through an explicit terminal cost.

Adjoint sampling chooses a memoryless base process by fixing a degenerate base drift $f_t :=0$, Dirac delta prior $\mu(x) := \delta_0(x)$, and $g(x)$ as above, leading to
\begin{subequations}
\label{eq: flat_SOC_AS}
\begin{align}
    &\min_{u \in \mathcal{U}} \ \mathbb{E}_{\boldsymbol{X} \sim p^u} \left[ \int_0^1 \frac12 \|u_t(X_t)\|^2 \diff t + \log \frac{p_1^{\mathrm{base}}(X_1)}{\nu(X_1)} \right], &\\
    &\text{s.t.} \quad \diff X_t = \sigma_t u_t (X_t) \diff t + \sigma_t \diff W_t, & \\
    &\hphantom{\text{s.t.}} \quad \ \ X_0 = 0. &
\end{align}
\end{subequations}
Observe \eqref{eq: flat_SB_SOC} differs from \eqref{eq: flat_SOC_AS} in the terminal cost and the relaxation of the source distribution from Dirac delta $X_0 = 0$ to a general source $\mu(X_0)$.
Nonetheless, \eqref{eq: memoryless_condition} severely restricts the choice of source distribution and base dynamics.
In practice, enforcing memorylessness typically requires either a degenerate $\mu$ or a strong noising mechanism that destroys useful dependence\footnote{For instance, in molecular Boltzmann sampling, a reasonable choice is to initialize $\mu$ as a Gaussian (harmonic approximation), to encode physically meaningful considerations.} between the initial and terminal states.

Since the base drift $f_t$ and the running cost $r_t$ are identically zero, the lean adjoint equation in \eqref{eq: ode_a_lean} simplifies significantly to $\tilde{a}_t = \nabla g(X_1), \ \forall t \in [0,1]$.
Therefore, this yields
\begin{equation}
\begin{split}
    &L_{\mathrm{AM}}(u) = \mathbb{E}_{p^{\bar{u}}} \left[ \frac{1}{2} \int_0^1 \| u_t(X_t) + \sigma_t \nabla g(X_1)\|^2 \diff t \right],\\
    &\boldsymbol{X} \sim p^{\bar{u}}, \quad \bar{u} = \texttt{stopgrad}(u).
    \label{eq: L_AM}
\end{split}
\end{equation}
Observe that for the general formulation \eqref{eq: flat_SOC}, AM demands two simulations, one to sample a trajectory from the stochastic process $\boldsymbol{X} \sim p^{\bar{u}}$ and one to solve the lean adjoint state backwards in time \eqref{eq: ode_a_lean} from a terminal condition at $t=1$ \eqref{eq: a_lean_terminal_cond}. However, in AS, no additional simulation of $\tilde{a}$ is required.
Essentially, here AM offers a simple interpretation: for each intermediate state $X_t$, simply regress the control onto the negative gradient of the terminal cost $-\nabla g(X_1)$ for all possible $X_1$ that can be reached from $X_t$. Since this is a moving target, this will, over the course of optimization, slowly shift the process towards regions with small terminal cost $g$.

Notwithstanding this, AM still involves two computationally expensive operations at every iteration. First, the simulation of the controlled process, and second, the evaluation of the terminal cost. The omission of both inefficiencies fundamentally motivates the development of \emph{Reciprocal Adjoint Matching}.

Since the training loss \eqref{eq: L_AM} only depends on endpoint-intermediate pairs $(X_t, X_1) \sim p_{t,1}^u$, rather than on the entire controlled trajectory $\boldsymbol{X}$, \eqref{eq: L_AM} is equivalent to sampling pairs $(X_t, X_1)$ from the joint distribution defined by $p_{t,1}^u$
\begin{equation}
    L_{\mathrm{AM}}(u) = \int_0^1 \mathbb{E}_{(X_t, X_1) \sim p_{t,1}^{\bar{u}}} \bigg[ \frac{1}{2} \| u_t(X_t) + \sigma_t \nabla g(X_1)\|^2 \bigg] \diff t .\notag
\end{equation}
However, sampling from $p_{t,1}^{\bar{u}}$ still requires simulating the controlled SDE. RAM avoids this by using the fact that, at the optimal solution $u^\star$, the conditional law of intermediate states is the one induced by the base process.
Hence, instead of sampling $(X_t, X_1)$ from the controlled trajectory law, RAM keeps the terminal marginal $p_1^{\bar{u}}$ generated by the current control and replaces the conditional $p_{t \vert 1}^{\bar{u}}$ with the base posterior $p_{t \vert 1}^{\mathrm{base}}$.
This operation is the Reciprocal projection:
\begin{equation}
    p^{\bar u}_{t,1}(X_t,X_1) \quad \leadsto \quad p^{\mathrm{base}}_{t|1}(X_t\mid X_1)p^{\bar u}_1(X_1).
\end{equation}
Equivalently, the current path measure is projected onto the SB that has the same terminal marginal but whose conditional bridges are those of the base process. 
This leads to the RAM objectives
\begin{equation}
    L_{\mathrm{RAM}} (u) = \int_0^1 \lambda_t \mathbb{E}_{X_t \sim p^{\mathrm{base}}_{t|1}(\cdot \vert X_1), \ X_1 \sim p_1^{\bar{u}}} \bigg[ \frac12 \|u_t(X_t) + \sigma_t \nabla g(X_1)\|^2 \bigg] \diff t, 
    \label{eq: L_RAM}
\end{equation}
where $\lambda_t := 1/\sigma_t^2$ applies a time scaling which does not affect the optimal solution but improves numerical stability.
That is, we sample $X_1$ according to the controlled process, then sample $X_t$ conditioned on $X_1$ using the posterior distribution defined by the base process. Assuming the base process is an SDE with zero drift (\eqref{eq: flat_controlled_sde} with $f_t, u \equiv 0$), and is initialized from a deterministic state $X_0 = x_0$, conditionals $p^{\mathrm{base}}_{t|1}$ are known in closed form and can be easily sampled\footnote{In Euclidean spaces, intermediate distributions are Gaussians.}. 

To further increase efficiency, we fix the regression target and delay updating it, and decouple $p(X_1)$ from the regression problem of learning $u$, delaying updates to $p(X_1)$ and performing multiple iterations to train $u$. This translates into the following algorithm: $(i)$ using the current control $u_i$ construct a buffer $\mathcal{B}= (X_1^{(i)}, \nabla g^{(i)})$ with samples $\{X_1^{(i)}\} \overset{\text{iid}}{\sim} p_1^u(X_1)$ and $\nabla g^{(i)} = \nabla g(X_1^{(i)})$; $(ii)$ obtain updated control $u_{i+1}$ by optimizing \eqref{eq: L_RAM} using samples $\{X_1^{(i)}, \nabla g^{(i)}\} \sim \mathcal{B}$.

\section{Schr\"{o}dinger Bridge Formulation and SOC Interpretation}
\label{sec: euclidean_SB_SOC}
The Schr\"{o}dinger Bridge problem \eqref{eq: SB_flat}, as an optimization problem with distribution constraints, is extensively studied in optimal transport, stochastic control and machine learning \citep{schrodinger1931umkehrung, leonard2014survey, chen2016relation}.
Formally, it is written as:
\begin{subequations}
\label{eq: SB_flat}
\begin{align}
     &\min_{u \in \mathcal{U}} \quad \left\{D_{\mathrm{KL}} (p^u \vert \vert p^{\mathrm{base}}) = \mathbb{E}_{\boldsymbol{X} \sim p^u} \left[ \int_0^1 \frac12 \|u_t (X_t) \|^2 \diff t \right]\right\},   \\
    &\ \text{s.t.} \quad \eqref{eq: flat_controlled_sde}, \quad X_1 \sim \nu.
\end{align}
\end{subequations}
where $D_{\mathrm{KL}}$ denotes the Kullback-Leibler divergence \citep{kullback1951information}.

Recall $p_{t \vert s}^{\mathrm{base}}(y \vert x) := p^{\mathrm{base}} (X_t =y \vert X_s = x)$ is the transition kernel of the base process for observing $y$ at time $t$ given $x$ at time $s$.
The SB potentials $\varphi_t (x),  \widehat{\varphi}_t (x) \in C^{1,2}([0,1] \times \R^d)$,
\begin{numcases}{}
    \varphi_t (x) = \int p_{1 \vert t}^{\mathrm{base}}(y \vert x) \varphi_1(y) \diff y, \qquad  \varphi_0 (x) \widehat{\varphi}_0 (x) = \mu(x) \label{eq: SB_cases_backward}\\
    \widehat{\varphi}_t (x) = \int p_{t \vert 0}^{\mathrm{base}}(x \vert y) \widehat{\varphi}_0(y) \diff y, \qquad  \varphi_1 (x) \widehat{\varphi}_1 (x) = \nu(x) \label{eq: SB_cases_forward}
\end{numcases}
are defined, up to some multiplicative constant, as solutions to forward and backward time integration with respect to $p_{t \vert s}^{\mathrm{base}}$.
These equations are computationally hard to solve due to the intractable integration and coupled boundaries at $t=0$ and $1$. Thus, following \citet[Theorem 3.1]{liu2025adjoint}, we recast the SB using a SOC interpretation
\begin{subequations}
\label{eq: flat_SB_SOC}
\begin{align}
    &\min_{u \in \mathcal{U}} \ \mathbb{E}_{\boldsymbol{X} \sim p^u} \left[ \int_0^1 \frac12 \|u_t(X_t)\|^2 \diff t + \log \frac{\widehat{\varphi}_1 (X_1)}{\nu(X_1)} \right], \qquad \ \text{s.t.} \quad \eqref{eq: flat_controlled_sde}.
\end{align}
\end{subequations}
Remarkably, the kinetic-optimal drift \eqref{eq: SB_u_star} in \eqref{eq: SB_flat} solves the associated SOC problem \eqref{eq: flat_SB_SOC}.
Since $\nu$ is given by \eqref{eq: pi_Rd}, $g(x) := -\log \varphi_1 =\log \frac{\widehat{\varphi}_1 (x)}{\nu(x)} = E(x) + \log \widehat{\varphi}_1 (x)$, dropping the irrelevant additive constant $\log Z$.

\paragraph{Hopf--Cole representation and optimal conditional law.}
Let $V_t:\R^d\to\R$ denote the value function associated with \eqref{eq: flat_SB_SOC}. It satisfies the Hamilton--Jacobi--Bellman
equation
\begin{equation}
    \partial_t V_t + \left\langle f_t,\nabla V_t\right\rangle + \frac{\sigma_t^2}{2}\Delta V_t - \frac{\sigma_t^2}{2}\|\nabla V_t\|^2 = 0,
    \qquad V_1=g.
    \label{eq: euclidean_HJB}
\end{equation}
Applying the Hopf--Cole transformation $\varphi_t(x):=\exp\bigl(-V_t(x)\bigr)$ linearizes \eqref{eq: euclidean_HJB} into the backward Kolmogorov
equation and gives the optimal control 
\begin{equation}
    u_t^\star(x) = -\sigma_t\nabla V_t(x) = \sigma_t\nabla\log\varphi_t(x).
    \label{eq: SB_u_star}
\end{equation}
Using \eqref{eq: SB_cases_backward}, the corresponding optimal conditional distribution is the Doob-transformed forward reference conditional
\begin{equation}
    p_{1|t}^\star(y|x)
    =
    p_{1|t}^{\mathrm{base}}(y|x)
    \frac{\varphi_1(y)}{\varphi_t(x)}.
    \label{eq: euclidean_optimal_conditional}
\end{equation}
Similarly, evaluating \eqref{eq: SB_cases_forward} at $t=1$ gives $\widehat\varphi_1(x) = \int_{\mathbb R^d} p_{1|0}^{\mathrm{base}}(x|y) \widehat\varphi_0(y)\,dy$.
The corresponding optimal reverse conditional density is
\begin{equation}
p_{0|1}^{\star}(y|x)
=
\frac{
\widehat\varphi_0(y)
p_{1|0}^{\mathrm{base}}(x|y)
}{
\widehat\varphi_1(x)
}.
\label{eq: euclidean_reverse_optimal_conditional}
\end{equation}
We now derive the two matching identities used by ASBS: denoising matching for the endpoint corrector $\nabla\log\widehat\varphi_1$, and adjoint matching for the controller score $\nabla\log\varphi_t$.

\colordm{Denoising matching (corrector)}.
Differentiating $\widehat\varphi_1$ with respect to $x$ gives
\begin{equation}
\nabla\log\widehat\varphi_1(x)
=
\frac{1}{\widehat\varphi_1(x)}
\int_{\mathbb R^d}
\nabla_x p_{1|0}^{\mathrm{base}}(x|y)
\widehat\varphi_0(y)\diff y
\stackrel{\eqref{eq: euclidean_reverse_optimal_conditional}}{=}
\mathbb E_{p_{0|1}^{\star}(\cdot|x)}
\left[
\nabla_x\log
p_{1|0}^{\mathrm{base}}(x|X_0)
\right].
\label{eq: euclidean_DM_identity}
\end{equation}
This is the exact denoising-matching characterization underlying
\eqref{eq: euclidean_corrector}. Note at the outset that this identity requires only a differentiable, strictly positive reference transition density and does not demand an additive or zero-drift reference, making denoising matching state-space agnostic.

\colortm{Adjoint matching (controller)}. 
For $f_t\equiv0$, the transition density is additive, i.e.,  $p_{1|t}^{\mathrm{base}}(y|x) = q_t(y-x)$ where $q_t = \mathcal N \left(0, \bar{\sigma}_t^2 I_d \right), \ \bar{\sigma}_t^2 := \int_t^1\sigma_s^2\diff s$.
Therefore,
\begin{align*}
    \nabla_x\varphi_t(x) &= \int_{\R^d} \nabla_xq_t(y-x)\varphi_1(y)\diff y \\
    &= - \int_{\R^d} \myellow{\nabla_yq_t(y-x)}\varphi_1(y)\diff y
    \shortnote{additivity: $\nabla_x q_t(y{-}x)=-\nabla_y q_t(y{-}x)$} \\
    &= \int_{\R^d} q_t(y-x)\myellow{\nabla_y\varphi_1(y)}\diff y.
    \shortnote{integration by parts}
\end{align*}
Dividing by $\varphi_t(x)$ and using
\eqref{eq: euclidean_optimal_conditional} yields
\begin{equation}
    \nabla\log\varphi_t(x)
    =
    \mathbb E_{p_{1|t}^\star(\cdot|x)}
    \left[
        \nabla\log\varphi_1(X_1)
    \right].
    \label{eq: euclidean_AM_identity}
\end{equation}
Unlike denoising matching, adjoint matching therefore relies on the additive structure of the reference transition kernel.

\section{Geometric and Stochastic Preliminaries}
\label{app: geometric_stochastic_preliminaries}

This section collects the definitions behind the operators $\nabla_{\mathcal M}$, $\operatorname{div}_{\mathcal M}$, $\Delta_{\mathcal M}$ of Section~\ref{subsec: notation}, and the exact meaning of the SDEs \eqref{eq: manifold_base_sde} and \eqref{eq: manifold_controlled_sde}. Throughout, $\mathcal M\subset\R^d$ is as in Section~\ref{subsec: notation}, equipped with the Riemannian metric induced by the Euclidean embedding, and $P_x$ denotes the orthogonal projection onto $T_x\mathcal M$, i.e., the tangential projection of \citet[Prop.~2.16]{lee2018introduction}.
Since $\mathcal M$ is a compact metric space, the path space $C([0,1],\mathcal M)$ is Polish, so regular conditional probabilities (disintegrations over endpoints or over $X_t$) exist and are unique up to null sets; we use them without further comment.

\paragraph{Intrinsic calculus.}
For $f\in C^1(\mathcal M)$, the Riemannian gradient $\nabla_{\mathcal M}f(x)\in T_x\mathcal M$ is the unique tangent vector satisfying
\begin{equation}
    \left\langle \nabla_{\mathcal M}f(x),\, v\right\rangle \;=\; df_x(v)
    \qquad \forall\, v\in T_x\mathcal M,
    \label{eq: riemannian_gradient_def}
\end{equation}
that is, $\nabla_{\mathcal M}f=(df)^{\sharp}$ \citep[Eq.~(2.14), p.~27]{lee2018introduction}. Because the metric is the induced one, \eqref{eq: riemannian_gradient_def} is equivalent to the extrinsic formula
\begin{equation*}
    \nabla_{\mathcal M}f(x) \;=\; P_x\,\nabla F(x)
    \qquad \text{for any smooth local extension } F \text{ of } f,
\end{equation*}
since $df_x = dF_x\vert_{T_x\mathcal M}$. 
For a smooth vector field $X$ on $\mathcal M$, the divergence $\operatorname{div}_{\mathcal M}X$ is characterized by $d\big(\iota_X\, d\mathrm{vol}_{\mathcal M}\big)=(\operatorname{div}_{\mathcal M}X)\, d\mathrm{vol}_{\mathcal M}$, with $\iota_X$ the interior product \citep[Eq.~(2.19)]{lee2018introduction} (via the Riemannian density \citep[Prop.~2.44]{lee2018introduction} in the unoriented case; coordinate expressions in \citep[Prop.~2.46]{lee2018introduction}), and the Laplace--Beltrami operator is $\Delta_{\mathcal M}:=\operatorname{div}_{\mathcal M}\circ\nabla_{\mathcal M}$ \citep[Eq.~(2.20)]{lee2018introduction}. Direct computation from these definitions yields the pointwise Leibniz identities
\begin{equation}
\begin{gathered}
    \nabla_{\mathcal M}(\varphi\psi)=\varphi\,\nabla_{\mathcal M}\psi+\psi\,\nabla_{\mathcal M}\varphi,
    \qquad
    \operatorname{div}_{\mathcal M}(h X)=h\operatorname{div}_{\mathcal M}X+\langle\nabla_{\mathcal M}h,\,X\rangle,\\
    \Delta_{\mathcal M}(\varphi\psi)=\varphi\,\Delta_{\mathcal M}\psi+2\big\langle\nabla_{\mathcal M}\varphi,\nabla_{\mathcal M}\psi\big\rangle+\psi\,\Delta_{\mathcal M}\varphi.
\end{gathered}
\label{eq: leibniz_identities}
\end{equation}
Since $\mathcal M$ is compact and boundaryless, the divergence theorem \citep[Problem~2-22, 2-23]{lee2018introduction} reads $\int_{\mathcal M}\operatorname{div}_{\mathcal M}X\,d\mathrm{vol}_{\mathcal M}=0$; applying it to $\psi X$ and to $\varphi\,\nabla_{\mathcal M}\psi$ gives the integration-by-parts and Green identities
\begin{equation*}
\begin{split}
    &\int_{\mathcal M}\big\langle\nabla_{\mathcal M}\psi,\,X\big\rangle\,d\mathrm{vol}_{\mathcal M}
    =-\int_{\mathcal M}\psi\,\operatorname{div}_{\mathcal M}X\,d\mathrm{vol}_{\mathcal M},
    \\
    &\int_{\mathcal M}\varphi\,\Delta_{\mathcal M}\psi\,d\mathrm{vol}_{\mathcal M}
    =\int_{\mathcal M}\psi\,\Delta_{\mathcal M}\varphi\,d\mathrm{vol}_{\mathcal M},
\end{split}
\end{equation*}
with no boundary terms. In particular, $\Delta_{\mathcal M}$ is symmetric on $L^2(d\mathrm{vol}_{\mathcal M})$, which underlies the symmetry of the heat kernel, $p_r^{\mathcal M}(x,y)=p_r^{\mathcal M}(y,x)$ \citep{grigoryan2009heat}. 

\paragraph{Diffusions on $\mathcal M$ as martingale problems.}
Given a tangent drift $b_t(x)\in T_x\mathcal M$, smooth in $(t,x)$, let $\mathcal L_t\psi:=\langle b_t,\nabla_{\mathcal M}\psi\rangle+\tfrac{\sigma_t^2}{2}\Delta_{\mathcal M}\psi$ for $\psi\in C^{\infty}(\mathcal M)$; the generators \eqref{eq: manifold_base_generator} and \eqref{eq: manifold_controlled_generator} are the cases $b_t=f_t^{\mathcal M}$ and $b_t=f_t^{\mathcal M}+\sigma_t u_t$, the tangency constraint $u_t(x)\in T_x\mathcal M$ being precisely what makes the controlled drift a vector field on $\mathcal M$. The SDE notations \eqref{eq: manifold_base_sde} and \eqref{eq: manifold_controlled_sde} are \emph{defined} through the associated martingale problem: an $\mathcal M$-valued process $X$ with $X_0\sim\mu_{\mathcal M}$ is an $\mathcal L$-diffusion if
\begin{equation}
    \psi(X_t)-\psi(X_0)-\int_0^t \mathcal L_s\psi(X_s)\,\diff s
    \quad\text{is a local martingale for every } \psi\in C^{\infty}(\mathcal M)
    \label{eq: manifold_martingale_problem}
\end{equation}
\citep[Def.~1.3.1]{hsu2002stochastic}. For smooth coefficients, an $\mathcal L$-diffusion exists \citep[Thm.~1.3.4]{hsu2002stochastic}, is unique in law \citep[Thm.~1.3.6]{hsu2002stochastic}, and is strong Markov \citep[Thm.~1.3.7]{hsu2002stochastic}; equivalently, it solves a Stratonovich SDE on $\mathcal M$, constructed by embedding and extension \citep[Def.~1.2.3, Thm.~1.2.9]{hsu2002stochastic} (see also \citep[Ch.~V]{ikeda1989stochastic}), with no explosion since $\mathcal M$ is compact. Time-dependent coefficients are covered by the standard space--time augmentation, i.e., applying the same results to the diffusion $(t,X_t)$ generated by $\partial_t+\mathcal L_t$ on $\R \times \mathcal M$. Brownian motion $W^{\mathcal M}$ is the $\tfrac12\Delta_{\mathcal M}$-diffusion \citep[Ch.~3]{hsu2002stochastic}, and its transition density is the heat kernel: smooth, strictly positive, and conservative, $\int_{\mathcal M}p_r^{\mathcal M}(x,y)\,d\mathrm{vol}_{\mathcal M}(y)=1$ \citep[\S4.1--4.2]{hsu2002stochastic}.

Given a strictly positive $h\in C^{1,2}([0,1]\times\mathcal M)$ that is \emph{space--time harmonic} for the reference dynamics,
\begin{equation}
    \partial_t h_t+\mathcal L^{\mathrm{base},\mathcal M}_t h_t=0,
    \label{eq: spacetime_harmonic}
\end{equation}
the \emph{Doob $h$-transform} of $p^{\mathrm{base}}_{\mathcal M}$ is the Markov law with transition densities $p^{\mathrm{base},\mathcal M}_{t|s}(y|x)\,h_t(y)/h_s(x)$ and Radon--Nikodym derivative $h_1(X_1)/h_0(X_0)$ on path space; see \citep[\S3]{leonard2014survey} and \citep{jamison1975markov}. By \eqref{eq: manifold_phi_backward} and the backward Kolmogorov equation for the reference kernel, the Schr\"odinger potential $\varphi^{\mathcal M}_t$ satisfies \eqref{eq: spacetime_harmonic}.

\paragraph{Geodesics and Transport.}
Throughout, $\nabla$ denotes the Levi--Civita connection of the induced metric and $\operatorname{Ric}^{\sharp}$ the Ricci endomorphism, $\langle\operatorname{Ric}^{\sharp}v,w\rangle=\operatorname{Ric}(v,w)$ \citep[Chs.~5, 7]{lee2018introduction}. Geodesics are the curves with $\nabla_{\dot\gamma}\dot\gamma=0$; since $\mathcal M$ is compact, it is geodesically complete, the exponential map $\operatorname{Exp}_x:T_x\mathcal M\to\mathcal M$ is globally defined, its local inverse $\operatorname{Exp}_x^{-1}$ (the Log map) is single-valued off the cut locus of $x$, and $d_{\mathcal M}$ denotes the geodesic distance \citep[Chs.~2, 4--6, 10]{lee2018introduction}. The parallel transport $\mathcal T^{\gamma}_{s\to t}:T_{\gamma_s}\mathcal M\to T_{\gamma_t}\mathcal M$ along a curve $\gamma$ is the solution operator of $\nabla_{\dot\gamma}A=0$ \citep[Ch.~4]{lee2018introduction}, i.e., precisely \eqref{eq: parallel_transport_equation}; a \emph{retraction} is a smooth map $R_x:T_x\mathcal M\to\mathcal M$ with $R_x(0)=x$ and $\diff R_x(0)=\mathrm{id}_{T_x\mathcal M}$, called second order when its curves match geodesics to second order, so that $R_x(v)=\operatorname{Exp}_x(v)+O(\|v\|^3)$. 
Along the paths of $X$, parallel transport is defined in the stochastic (Stratonovich) sense through the horizontal lift on the frame bundle, and every $\mathcal M$-valued semimartingale admits an $\R^n$-valued \emph{anti-development} $W$ \citep[\S2.3]{hsu2002stochastic}; $X$ is a Brownian motion on $\mathcal M$ exactly when its anti-development is a Euclidean Brownian motion \citep[Ch.~3]{hsu2002stochastic}. In this language, the damped transport \eqref{eq: damped_transport} means $\mathcal W_{t,s}=/\!/_{t,s}\,Q_{t,s}$, where $/\!/_{t,s}$ is stochastic parallel transport and $Q_{t,s}$ solves the pathwise linear ODE $\frac{\diff}{\diff s}Q_{t,s}=-\frac{\sigma_s^2}{2}(/\!/_{t,s})^{-1}\operatorname{Ric}^{\sharp}_{X_s}/\!/_{t,s}\,Q_{t,s}$, $Q_{t,t}=I$; in particular, by Gr\"onwall's inequality, $\|\mathcal W_{t,s}\|_{\mathrm{op}}\le\exp\big(\tfrac12\int_t^s\sigma_r^2\,\|\operatorname{Ric}^{-}\|_\infty\,\diff r\big)$, so damped transport differs from ordinary parallel transport by a bounded finite-variation correction on short time intervals.

\section{Proofs}
\label{app: manifold_SOC_formulation}

\subsection[Proof of the Manifold SOC--SB Equivalence]{Proof of \texorpdfstring{Theorem~\ref{theo: SOC_SB_manifold_equivalence}}{Theorem 1} (Manifold SOC--SB Equivalence)}
\begin{proof}[Proof of Theo.~\ref{theo: SOC_SB_manifold_equivalence}]
\textbf{Step 1 (Feasibility and Schr\"odinger potentials).}$\quad$
Let $R:=p_{\mathcal M}^{\mathrm{base}}$ and let $R_{0,1}$ be its endpoint law. Under the stated positivity and smoothness assumptions, the static Schr\"odinger problem is feasible with finite entropy: on the compact manifold\footnote{The product coupling $\mu_{\mathcal M}\otimes\nu_{\mathcal M}$ is the independent endpoint law on $\mathcal M\times\mathcal M$, with density $\mu_{\mathcal M}(x_0)\nu_{\mathcal M}(x_1)$ and marginals $\mu_{\mathcal M}$ and $\nu_{\mathcal M}$. It is used only as a finite-entropy feasible coupling, not as the Schr\"odinger optimizer.} $\mu_{\mathcal M}\otimes\nu_{\mathcal M}$ has a bounded, strictly positive density with respect to $R_{0,1}$, so $D_{\mathrm{KL}}(\mu_{\mathcal M}\otimes\nu_{\mathcal M}\,\|\,R_{0,1})<\infty$. Hence the dynamic Schr\"odinger problem has a unique minimizer $P^\star:=p_{\mathcal M}^\star$, and the Schr\"odinger system admits positive potentials satisfying \eqref{eq: manifold_phi_backward}--\eqref{eq: manifold_phihat_forward} and the endpoint products in the theorem \citep{jamison1975markov, follmer1988random, leonard2014survey}.

\textbf{Step 2 (Regularity and admissibility of $u^\star$).}$\quad$
We first record the regularity used below. Since the reference kernels are smooth and strictly positive on the compact manifold, the representations \eqref{eq: manifold_phi_backward}--\eqref{eq: manifold_phihat_forward} show that $\varphi_t^{\mathcal M}$ is strictly positive and smooth in $(t,x)$ for $t\in[0,1)$ and $\widehat\varphi_t^{\mathcal M}$ for $t\in(0,1]$, with boundary values $\varphi_1^{\mathcal M}=\nu_{\mathcal M}/\widehat\varphi_1^{\mathcal M}$ and $\widehat\varphi_0^{\mathcal M}=\mu_{\mathcal M}/\varphi_0^{\mathcal M}$ that are strictly positive and $C^2$ by the assumptions on $\mu_{\mathcal M},\nu_{\mathcal M}$; consequently both potentials are $C^{1,2}$ on all of $[0,1]\times\mathcal M$. In particular, every gradient appearing below is well defined, and $u_t^\star=\sigma_t\nabla_{\mathcal M}\log\varphi_t^{\mathcal M}$ is a continuous, bounded tangent field on the compact set $[0,1]\times\mathcal M$, hence admissible: $u^\star\in\mathcal U_{\mathcal M}$.

\textbf{Step 3 (Static optimizer and endpoint factorization).}$\quad$
Since $R_0=\mu_{\mathcal M}$, the static optimizer has density $\pi^\star(x_0,x_1)=\widehat\varphi_0^{\mathcal M}(x_0)\,p_{1|0}^{\mathrm{base},\mathcal M}(x_1|x_0)\,\varphi_1^{\mathcal M}(x_1)$ with respect to $d\mathrm{vol}_{\mathcal M}\otimes d\mathrm{vol}_{\mathcal M}$. Indeed, its marginals satisfy, by \eqref{eq: manifold_phi_backward}--\eqref{eq: manifold_phihat_forward} and the boundary products,
\begin{align*}
&\int_{\mathcal M}\pi^\star(x_0,x_1)\,d\mathrm{vol}_{\mathcal M}(x_1)
=\widehat\varphi_0^{\mathcal M}(x_0)\varphi_0^{\mathcal M}(x_0)=\mu_{\mathcal M}(x_0),
\shortnote{by \eqref{eq: manifold_phi_backward}}
\\
&\int_{\mathcal M}\pi^\star(x_0,x_1)\,d\mathrm{vol}_{\mathcal M}(x_0)
=\varphi_1^{\mathcal M}(x_1)\widehat\varphi_1^{\mathcal M}(x_1)=\nu_{\mathcal M}(x_1),
\shortnote{by \eqref{eq: manifold_phihat_forward}}
\end{align*}
so both endpoint constraints hold. Dividing by the reference endpoint density $\mu_{\mathcal M}(x_0)\,p_{1|0}^{\mathrm{base},\mathcal M}(x_1|x_0)$ yields the endpoint-factor representation
\begin{equation}
\frac{dP^\star}{dR}
=
\frac{\widehat\varphi_0^{\mathcal M}(X_0)}{\mu_{\mathcal M}(X_0)}\varphi_1^{\mathcal M}(X_1).
\label{eq: theorem_proof_endpoint_factor}
\end{equation}
The initial Schr\"odinger constraint gives $\widehat\varphi_0^{\mathcal M}/\mu_{\mathcal M}=1/\varphi_0^{\mathcal M}$, proving \eqref{eq: manifold_SB_RN}.

\textbf{Step 4 (Doob $h$-transform and the optimal control).}$\quad$
The corresponding Doob transform has transition density
\begin{equation*}
p_{t|s}^{\star,\mathcal M}(y|x)
=
p_{t|s}^{\mathrm{base},\mathcal M}(y|x)
\frac{\varphi_t^{\mathcal M}(y)}{\varphi_s^{\mathcal M}(x)}.
\end{equation*}
This kernel is a bona fide Markov transition density: it is nonnegative, and by \eqref{eq: manifold_phi_backward} together with the Chapman--Kolmogorov identity,
$\int_{\mathcal M}p_{t|s}^{\mathrm{base},\mathcal M}(y|x)\,\varphi_t^{\mathcal M}(y)\,d\mathrm{vol}_{\mathcal M}(y)=\varphi_s^{\mathcal M}(x)$,
so it integrates to one. 
Moreover, $\varphi^{\mathcal M}$ is space--time harmonic, \eqref{eq: spacetime_harmonic}, and the endpoint tilt \eqref{eq: theorem_proof_endpoint_factor} is exactly the Doob $h$-transform of the Markov law $R$ by $\varphi^{\mathcal M}$ (Appendix~\ref{app: geometric_stochastic_preliminaries}), whence $P^\star$ is Markov with the displayed transitions.
For a smooth test function $\psi$, its generator is
\begin{align*}
\mathcal L_t^{\star,\mathcal M}\psi
&=
\frac{1}{\varphi_t^{\mathcal M}}
\left[
\mathcal L_t^{\mathrm{base},\mathcal M}(\varphi_t^{\mathcal M}\psi)
-
\psi\mathcal L_t^{\mathrm{base},\mathcal M}\varphi_t^{\mathcal M}
\right]\\
&=
\mathcal L_t^{\mathrm{base},\mathcal M}\psi
+
\myellow{\sigma_t^2
\left\langle
\nabla_{\mathcal M}\log\varphi_t^{\mathcal M},
\nabla_{\mathcal M}\psi
\right\rangle},
\shortnote{by \eqref{eq: leibniz_identities}}
\end{align*}
Since $\myellow{\nabla_{\mathcal M}\log\varphi_t^{\mathcal M}(x)}\in T_x\mathcal M$, comparing this generator with \eqref{eq: manifold_controlled_generator} shows that $P^\star$ and the diffusion controlled by \eqref{eq: manifold_optimal_control} solve the same well-posed martingale problem \eqref{eq: manifold_martingale_problem} with initial law $\mu_{\mathcal M}$; they coincide by uniqueness in law (Appendix~\ref{app: geometric_stochastic_preliminaries}).

\textbf{Step 5 (Girsanov entropy identity).}$\quad$
Let $B_t$ denote the $\R^n$-valued Brownian anti-development of the
reference diffusion, and let $U_t:\R^n\to T_{X_t}\mathcal M$ be the
associated orthonormal stochastic frame. Define
$\widetilde u_t:=U_t^{-1}u_t(X_t)$.
Then, by Girsanov's theorem,
\begin{equation}
\frac{\diff p_{\mathcal M}^{u}}{\diff p_{\mathcal M}^{\mathrm{base}}}
=
\exp\!\left(
\int_0^1\langle\widetilde u_t,\diff B_t\rangle
-\frac12\int_0^1\|\widetilde u_t\|^2\diff t
\right).
\label{eq: manifold_girsanov_density}
\end{equation}
Taking $\mathbb E_{p^u_{\mathcal M}}\log$ of \eqref{eq: manifold_girsanov_density} yields\footnote{For the Euclidean entropy--energy identity with proof, see
\citep[Prop.~4.1, Eq.~(4.7)]{leonard2014survey}. In the manifold
setting, the same argument applies after representing the tangent
control in the orthonormal stochastic frame, so that the
anti-development is an $\mathbb R^n$-valued Brownian motion; see
\citep[\S2.3, Ch.~3]{hsu2002stochastic}.} the entropy identity
\begin{equation}
D_{\mathrm{KL}}(p_{\mathcal M}^u\|R)
=
\frac12\mathbb E_{p_{\mathcal M}^u}
\int_0^1\|u_t(X_t)\|_{T_{X_t}\mathcal M}^2\diff t.
\label{eq: manifold_girsanov_entropy}
\end{equation}
\textbf{Step 6 (SOC identification and terminal adjoint).}$\quad$
To identify the terminal-cost SOC minimizer as a path measure, let $g:\mathcal M\to\mathbb R$ and define
\begin{equation*}
Z_g(x):=\mathbb E_R[e^{-g(X_1)}\mid X_0=x],
\qquad
\frac{dQ^g}{dR}:=\frac{e^{-g(X_1)}}{Z_g(X_0)}.
\end{equation*}
For every path measure $P\ll R$ with $P_0=\mu_{\mathcal M}$,
\begin{align}
D_{\mathrm{KL}}(P\|R)+\mathbb E_P[g(X_1)]
&=
D_{\mathrm{KL}}(P\|Q^g)
-
\mathbb E_{\mu_{\mathcal M}}[\log Z_g(X_0)].
\label{eq: gibbs_path_decomposition}
\end{align}
Hence $Q^g$ is the unique minimizer. Choose
\begin{equation*}
\myellow{g^{\star,\mathcal M}(x)
=-\log\varphi_1^{\mathcal M}(x)}
=
\log\frac{\widehat\varphi_1^{\mathcal M}(x)}{\nu_{\mathcal M}(x)}.
\end{equation*}
Then $Z_{g^\star}(x)=\varphi_0^{\mathcal M}(x)$ by \eqref{eq: manifold_phi_backward}, and therefore
\begin{equation*}
\frac{dQ^{g^\star}}{dR}
=
\frac{\myellow{\varphi_1^{\mathcal M}(X_1)}}{\varphi_0^{\mathcal M}(X_0)}
=
\frac{dP^\star}{dR}.
\shortnote{by \eqref{eq: manifold_SB_RN}}
\end{equation*}
Thus the SOC and SB solutions coincide as complete path measures. Moreover, the unconstrained minimizer is attained inside the admissible class: the control $u^\star=\sigma_t\nabla_{\mathcal M}\log\varphi_t^{\mathcal M}$ belongs to $\mathcal U_{\mathcal M}$ (as recorded at the start of the proof) and induces precisely $P^\star=Q^{g^\star}$. Combining \eqref{eq: manifold_girsanov_entropy} with \eqref{eq: gibbs_path_decomposition} then proves \eqref{eq: SOC_theo1} over the admissible control class.

Finally, if $\nu_{\mathcal M}\propto e^{-E}$, differentiating $g^{\star,\mathcal M}=E+\log\widehat\varphi_1^{\mathcal M}+\mathrm{const.}$ proves \eqref{eq: manifold_terminal_adjoint}.
\end{proof}

\subsection{Reciprocal Property}
A fundamental result in SB theory, which follows from Theorem~\ref{theo: SOC_SB_manifold_equivalence}, known as the reciprocal property, states that the optimal controlled bridge conditioned on its endpoints is identical to the uncontrolled reference bridge \citep{leonard2014survey}.
\begin{corollary}[Reciprocal property]\label{cor: reciprocal_property_manifold}
Let $p_{\mathcal M}^\star$ solve \eqref{eq: manifold_SB_problem}. Then, for every $t\in(0,1)$,
\begin{equation}
p_{\mathcal M}^\star(X_t\mid X_0,X_1)
=
p_{\mathcal M}^{\mathrm{base}}(X_t\mid X_0,X_1)
\end{equation}
for $p_{0,1}^{\star,\mathcal M}$-almost every endpoint pair.
\end{corollary}
\begin{proof}
By \eqref{eq: manifold_SB_RN}, the optimal law is an endpoint tilt of the Markov reference law:
\begin{equation*}
\frac{d p_{\mathcal M}^\star}{d p_{\mathcal M}^{\mathrm{base}}}(\boldsymbol X)
=
\frac{\varphi_1^{\mathcal M}(X_1)}{\varphi_0^{\mathcal M}(X_0)}.
\end{equation*}
Conditioning on $X_0=x_0$ and $X_1=x_1$ turns this Radon--Nikodym factor into a constant, which cancels under conditional normalization. Therefore the complete endpoint-conditioned path laws coincide. In particular, for $t\in(0,1)$,
\begin{equation*}
p_{\mathcal M}^\star(x_t\mid x_0,x_1)
=
\frac{p_{t|0}^{\mathrm{base},\mathcal M}(x_t|x_0)
      p_{1|t}^{\mathrm{base},\mathcal M}(x_1|x_t)}
     {p_{1|0}^{\mathrm{base},\mathcal M}(x_1|x_0)}
=
p_{\mathcal M}^{\mathrm{base}}(x_t\mid x_0,x_1),
\end{equation*}
where the second equality is the Markov factorization of the reference bridge, obtained from the Markov property and the Chapman--Kolmogorov identity, and the first equality then follows from the coincidence of the endpoint-conditioned path laws established above.
\end{proof}
Corollary~\ref{cor: reciprocal_property_manifold} is a path-space consequence of the endpoint factorization \eqref{eq: manifold_SB_RN}, not of Euclidean geometry. It permits exact reference-bridge sampling when that bridge is available, and motivates the noisy-geodesic approximation when it is not.

\subsection{Exact Intrinsic Matching Identities}
\graybox{%
\begin{lemma}[Exact manifold denoising matching for the corrector]
\label{prop: exact_manifold_DM}
Under the hypotheses of Theorem~\ref{theo: SOC_SB_manifold_equivalence},
\begin{equation}
\nabla_{\mathcal M}\log\widehat\varphi_1^{\mathcal M}(x)
=
\mathbb E_{p_{0|1}^{\star,\mathcal M}(\cdot|x)}
\left[
\nabla_{\mathcal M,x}\log p_{1|0}^{\mathrm{base},\mathcal M}(x|X_0)
\right].
\label{eq: manifold_corrector_DM}
\end{equation}
Consequently, the corrector admits the exact variational characterization
\begin{equation}
\nabla_{\mathcal M}\log\widehat\varphi_1^{\mathcal M}
=
\argmin_h\mathbb E_{p_{0,1}^{\star,\mathcal M}}
\left[
\left\|h(X_1)-\nabla_{\mathcal M,X_1}\log p_{1|0}^{\mathrm{base},\mathcal M}(X_1|X_0)\right\|^2
\right].
\label{eq: manifold_exact_corrector_loss}
\end{equation}
\end{lemma}
}
\begin{proof}
Evaluating \eqref{eq: manifold_phihat_forward} at $t=1$ and
differentiating under the integral sign gives
\begin{align*}
\nabla_{\mathcal M}\log\widehat\varphi_1^{\mathcal M}(x)
&=
\frac{1}{\widehat\varphi_1^{\mathcal M}(x)}
\int_{\mathcal M}
\nabla_{\mathcal M,x}
p_{1|0}^{\mathrm{base},\mathcal M}(x|y)
\widehat\varphi_0^{\mathcal M}(y)
\,d\mathrm{vol}_{\mathcal M}(y)
\shortnote{by \eqref{eq: manifold_phihat_forward}}
\\
&=
\int_{\mathcal M}
\nabla_{\mathcal M,x}
\log p_{1|0}^{\mathrm{base},\mathcal M}(x|y)
\myellow{
\frac{
p_{1|0}^{\mathrm{base},\mathcal M}(x|y)
\widehat\varphi_0^{\mathcal M}(y)
}{
\widehat\varphi_1^{\mathcal M}(x)
}}
\,d\mathrm{vol}_{\mathcal M}(y).
\end{align*}
Since the optimal reverse conditional density is $\myellow{p_{0|1}^{\star,\mathcal M}(y|x)} = \frac{ \widehat\varphi_0^{\mathcal M}(y) p_{1|0}^{\mathrm{base},\mathcal M}(x|y)}{\widehat\varphi_1^{\mathcal M}(x)}$, the preceding identity is exactly \eqref{eq: manifold_corrector_DM}.
Finally, conditional expectation is the unique squared-loss minimizer
up to $p_1^{\star,\mathcal M}$-almost-everywhere equality, which yields
\eqref{eq: manifold_exact_corrector_loss}.
\end{proof}
\graybox{%
\begin{lemma}[Exact manifold adjoint matching for the controller]
\label{prop: exact_manifold_AM}
Assume the hypotheses of Theorem~\ref{theo: SOC_SB_manifold_equivalence}, $f_t^{\mathcal M}\equiv0$, and $\sigma_t = \sigma$ for simplicity. Then for $t\in[0,1)$,
\begin{equation}
\nabla_{\mathcal M}\log\varphi_t^{\mathcal M}(x)
=
\mathbb E_{p_{\mathcal M}^{\star}(\cdot\mid X_t=x)}
\left[
(\mathcal W_{t,1})^*
\nabla_{\mathcal M}\log\varphi_1^{\mathcal M}(X_1)
\right].
\label{eq: manifold_exact_controller_AM}
\end{equation}
Equivalently, defining $a_1(x):=\nabla_{\mathcal M}E(x)+\nabla_{\mathcal M}\log\widehat\varphi_1^{\mathcal M}(x)$, the optimal controller satisfies
\begin{equation}
u^\star
=
\argmin_u
\int_0^1\mathbb E_{p_{\mathcal M}^{\star}}
\left[
\left\|u_t(X_t)+\sigma_t(\mathcal W_{t,1})^*a_1(X_1)\right\|_{T_{X_t}\mathcal M}^2
\right]\diff t.
\label{eq: manifold_exact_AM_loss}
\end{equation}
Here, for each $t$, the transported adjoint $(\mathcal W_{t,1})^*a_1(X_1)$ is a functional of the path segment on $[t,1]$.
\end{lemma}
}
\begin{proof}
Let $P_{t,s}$ denote the heat semigroup of the zero-drift reference diffusion. The derivative-semigroup formula associated with \eqref{eq: damped_transport} \citep{thalmaier1998remarks, coulibaly2011brownian} gives, for every smooth $F$,
\begin{equation}
\nabla_{\mathcal M}P_{t,s}F(x)
=
\mathbb E_{p_{\mathcal M}^{\mathrm{base}}(\cdot\mid X_t=x)}
\left[(\mathcal W_{t,s})^*\nabla_{\mathcal M}F(X_s)\right].
\label{eq: derivative_semigroup_formula}
\end{equation}
Since $\varphi_t^{\mathcal M}=P_{t,1}\varphi_1^{\mathcal M}$,
\begin{align*}
\nabla_{\mathcal M}\varphi_t^{\mathcal M}(x)
&=
\mathbb E_{p_{\mathcal M}^{\mathrm{base}}(\cdot\mid X_t=x)}
\left[(\mathcal W_{t,1})^*\nabla_{\mathcal M}\varphi_1^{\mathcal M}(X_1)\right]
\shortnote{by \eqref{eq: derivative_semigroup_formula}}\\
&=
\mathbb E_{p_{\mathcal M}^{\mathrm{base}}(\cdot\mid X_t=x)}
\left[\varphi_1^{\mathcal M}(X_1)(\mathcal W_{t,1})^*\nabla_{\mathcal M}\log\varphi_1^{\mathcal M}(X_1)\right].
\shortnote{$\nabla\varphi=\varphi\nabla\log\varphi$}
\end{align*}
By \eqref{eq: manifold_SB_RN}, $dp_{\mathcal M}^{\star}/dp_{\mathcal M}^{\mathrm{base}}=\varphi_1^{\mathcal M}(X_1)/\varphi_0^{\mathcal M}(X_0)$. Hence, for every event $A$ in the future $\sigma$-algebra $\sigma(X_s,\,s\ge t)$,
\begin{equation*}
p_{\mathcal M}^{\star}(A\mid X_t=x)
=
\frac{\mathbb E_{R}\!\left[\mathbf 1_A\,\varphi_1^{\mathcal M}(X_1)\,\myellow{\varphi_0^{\mathcal M}(X_0)^{-1}}\mid X_t=x\right]}
     {\mathbb E_{R}\!\left[\varphi_1^{\mathcal M}(X_1)\,\myellow{\varphi_0^{\mathcal M}(X_0)^{-1}}\mid X_t=x\right]}
=
\frac{\mathbb E_{R}\!\left[\mathbf 1_A\,\varphi_1^{\mathcal M}(X_1)\mid X_t=x\right]}{\varphi_t^{\mathcal M}(x)},
\end{equation*}
where the second equality holds because, by the Markov property of $R:=p_{\mathcal M}^{\mathrm{base}}$, the past factor $\myellow{\varphi_0^{\mathcal M}(X_0)^{-1}}$ is conditionally independent of $\bigl(\mathbf 1_A,\varphi_1^{\mathcal M}(X_1)\bigr)$ given $X_t$, so it factors out of numerator and denominator and cancels, while $\mathbb E_R[\varphi_1^{\mathcal M}(X_1)\mid X_t=x]=\varphi_t^{\mathcal M}(x)$ by \eqref{eq: manifold_phi_backward}. Thus the restriction of the optimal future-path law conditional on $X_t=x$ is the terminal tilt
\begin{equation*}
\frac{d p_{\mathcal M}^{\star}(\cdot\mid X_t=x)}{d p_{\mathcal M}^{\mathrm{base}}(\cdot\mid X_t=x)}
=
\frac{\varphi_1^{\mathcal M}(X_1)}{\varphi_t^{\mathcal M}(x)}.
\end{equation*}
Since the Radon--Nikodym density depends only on $X_1$ while the transported integrand is a measurable functional of the path segment on $[t,1]$, dividing the preceding gradient identity by $\varphi_t^{\mathcal M}(x)$ proves \eqref{eq: manifold_exact_controller_AM}.
Moreover, the terminal Schr\"odinger relation
$\varphi_1^{\mathcal M}\widehat\varphi_1^{\mathcal M}
=\nu_{\mathcal M}\propto e^{-E}$ implies
\begin{equation*}
\nabla_{\mathcal M}\log\varphi_1^{\mathcal M}
=
-\nabla_{\mathcal M}E
-
\nabla_{\mathcal M}\log\widehat\varphi_1^{\mathcal M}
=
-a_1.
\end{equation*}
Combining this identity with
\eqref{eq: manifold_optimal_control} and
\eqref{eq: manifold_exact_controller_AM} gives $u_t^\star(x) = -\sigma_t \mathbb E_{p_{\mathcal M}^{\star}(\cdot\mid X_t=x)} \left[ (\mathcal W_{t,1})^*a_1(X_1) \right]$.

Finally, conditional expectation is the unique squared-loss minimizer
up to $dt\otimes p_t^{\star,\mathcal M}$-almost-everywhere equality,
which gives \eqref{eq: manifold_exact_AM_loss}.
\end{proof}
In Euclidean space, $\operatorname{Ric}^{\sharp}=0$ and all tangent spaces are canonically identified, so $(\mathcal W_{t,1})^*=I$ and \eqref{eq: manifold_exact_controller_AM} reduces to \eqref{eq: euclidean_AM_identity}. 

\section{Additional Theoretical Analysis}
\label{app: additional_theoretical_results}
\subsection{Algorithmic Structures}
\label{app: alg_structure}
\begin{algorithm}
\caption{R--ASBS}\label{alg: asbs_m}
\begin{algorithmic}[1]
\Require energy $E(x)$, source $\mu_{\mathcal{M}}$, retraction map $R_x$, orthogonal projector $P_x$, batch size $B$, training epochs $K$, SDE steps $N$, step size $\Delta t$, noise amplitude $\sigma$.
\State \textbf{Initialize:} Controller network $u_\theta(x,t)$ and corrector network $h_\phi(x)$ with weights $\theta_0,\phi_0$.
\For{$k=1$ \textbf{to} $K$}
    \hspace*{-\fboxsep}\colorbox{lightgraybox}{\parbox{\dimexpr\linewidth-\fboxsep\relax}{%
    \State Sample initial particles $\{X_0^{(i)}\}_{i=1}^B\sim\mu_{\mathcal M}$ \COMMENT{controller matching$~$}
    \FORR{$n=0$ \textbf{to} $N-1$ \textbf{do}}
        \State Sample ambient noise $\epsilon_n^{(i)}\sim\mathcal N(0,I_d)$
        \State $X_{n+1}^{(i)}\gets R_{X_n^{(i)}}\!\left(\sigma P_{X_n^{(i)}}u_\theta^{(k-1)}(X_n^{(i)},n\Delta t)\Delta t+\sigma\sqrt{\Delta t}P_{X_n^{(i)}}\epsilon_n^{(i)}\right)$
    \ENDFORR
    \State Sample intermediate times $t^{(i)}\sim\mathrm{Unif}[0,1]$
    \State Compute the base geodesic $\gamma_t$ from $X_0$ to $X_N$
    \State Sample tangent noise $\epsilon_b^{(i)}\sim\mathcal N(0,I_{T_\gamma})$
    \State $X_{t^{(i)}}^{(i)}\gets R_{\gamma^{(i)}_{t^{(i)}}}\!\left(\sigma\sqrt{t^{(i)}(1-t^{(i)})}\epsilon_b^{(i)}\right)$
    \State Evaluate terminal adjoint: $a_1^{(i)}\gets\nabla_{\mathcal M}E(X_N^{(i)})+P_{X_N^{(i)}}h_\phi^{(k-1)}(X_N^{(i)})$
    \State Compute $a_t^{(i)}$ by transporting $a_1^{(i)}$ along the chosen geodesic from $X_N^{(i)}$ to $X_{t^{(i)}}^{(i)}$, as in \eqref{eq: practical_transport_approximation}
    \State Update $\theta$ by descending $\nabla_\theta\mathcal L_u(\theta)$, where
    $\mathcal L_u=\frac1B\sum_{i=1}^B\left\|P_{X_{t^{(i)}}^{(i)}}u_\theta(X_{t^{(i)}}^{(i)},t^{(i)})+\sigma a_t^{(i)}\right\|^2$
    }}\vspace{0.5em}
    \hspace*{-\fboxsep}\colorbox{lightgraybox}{\parbox{\dimexpr\linewidth-\fboxsep\relax}{%
    \State Re-sample trajectories with updated controller: $\{(X_0^{(i)},X_N^{(i)})\}_{i=1}^B\sim p_{0,1}^{u^{(k)}}$
    \State Compute approximate denoising-corrector targets $b^{(i)}$ using \eqref{eq: Varadhan_approximation} \COMMENT{corrector matching$~$}
    \State Update $\phi$ by descending $\nabla_\phi\mathcal L_h(\phi)$, where
    $\mathcal L_h=\frac1B\sum_{i=1}^B\left\|P_{X_N^{(i)}}h_\phi(X_N^{(i)})-b^{(i)}\right\|^2$
    }}
\EndFor
\State \Return Trained controller $u_\theta(x,t)$
\end{algorithmic}
\end{algorithm}
For clarity, Algorithm~\ref{alg: asbs_m} is written for a constant noise amplitude $\sigma$; the time-varying case follows by replacing $\sigma^2$ with the accumulated variance $\int_s^t\sigma_r^2\diff r$.

\begin{algorithm}
\caption{Extended R--ASBS.}\label{alg: extended_asbs_m} 
\begin{algorithmic}[1] 
\Require Same inputs as in Algorithm~\ref{alg: asbs_m}.
\State \textbf{Initialize:} Controller network $u_\theta(x, t)$ and corrector network $h_\phi(x)$ with weights $\theta_0, \phi_0$.
\For{$k = 1$ \textbf{to} $K$}
    \hspace*{-\fboxsep}\colorbox{lightgraybox}{\parbox{\dimexpr\linewidth-\fboxsep\relax}{%
    \State Sample initial particles $\{X_0^{(i)}\}_{i=1}^B \sim \mu_\mathcal{M}$ \COMMENT{controller matching$~$}
    \FORR{$n = 0$ \textbf{to} $N-1$ \textbf{do}}
         \State Sample ambient noise $\epsilon_n^{(i)} \sim \mathcal{N}(0, I_d)$.
         \State $X_{n+1}^{(i)} \gets \Pi_{\mathcal{M}} \Big( X_n^{(i)} + \sigma P_{X_n^{(i)}} u_\theta^{(k-1)} (X_n^{(i)}, n\Delta t)\Delta t + \sigma \sqrt{\Delta t} P_{X_n^{(i)}} \epsilon_n^{(i)} \Big)$.
    \ENDFORR
    \State Initialize $v_N^{(i)} = P_{X_N^{(i)}}\big(\nabla E(X_N^{(i)})+h_\phi^{(k-1)}(X_N^{(i)})\big)$
    \FORR{$j = N-1, N-2, \dots, 0$ \textbf{do}}
        \State $v_j^{(i)} \gets P_{X_j^{(i)}}v_{j+1}^{(i)}$ \COMMENT{PAT \eqref{eq: PAT}}
    \ENDFORR
    \State Update $u_\theta^{(k)} \gets  \argmin_{u} \frac{1}{B} \sum_{i=1}^B \sum_{j=0}^{N-1} \| P_{X_j^{(i)}} u_\theta(X_j^{(i)}, t_j) + \sigma v_j^{(i)} \|^2$
    }}\vspace{0.5em}
    \hspace*{-\fboxsep}\colorbox{lightgraybox}{\parbox{\dimexpr\linewidth-\fboxsep\relax}{%
    \State Re-sample $\{(X_0^{(i)},X_N^{(i)})\}_{i=1}^B \sim p_{0,1}^{u^{(k)}}$ \COMMENT{corrector matching$~$}
    \State Compute projected chord: $b^{(i)} \gets - P_{X_N^{(i)}} \frac{X_N^{(i)} - X_0^{(i)}}{\sigma^2} \in T_{X_N^{(i)}} \mathcal{M}$
    \State Update $h_\phi^{(k)} \gets \argmin_{h} \frac{1}{B} \sum_{i=1}^B \|P_{X_N^{(i)}} h_\phi(X_N^{(i)}) - b^{(i)}\|^2$
    }}
 \EndFor
 \State \Return Trained controller $u_\theta(x, t)$
\end{algorithmic}
\end{algorithm}

\subsection{Controller Normality Independence}
\label{app: theoretical_analysis}
In implementations based on the embedded representation, an unconstrained ambient vector field $\widetilde u_t:\mathcal M\to\mathbb R^d$ can be made admissible by the tangent projection $u_t(x) = P_x\widetilde u_t(x)$.
Proposition~\ref{prop: control_independence_normal} proves the normal component is dynamically irrelevant, and the controller is identifiable only modulo the normal bundle, namely $\widetilde{u} \sim  \widetilde{u} + n$, $n(x) \in N_x \mathcal M$.
\begin{proposition}
\label{prop: control_independence_normal}
Let $\mathcal M\subset\mathbb R^d$ satisfy the hypotheses in Section~\ref{subsec: notation}. Let $\widetilde u_1,\widetilde u_2:[0,1]\times\mathcal M\to\mathbb R^d$ be ambient control fields whose projected tangent components coincide, i.e., $P_x\widetilde u_1(t,x)=P_x\widetilde u_2(t,x),\ \forall (t,x) \in [0,1]\times\mathcal M$.
Assume the corresponding martingale problems are well posed. Then the intrinsic controlled diffusions obtained by inserting $P_x\widetilde u_i(t,x)$ into the drift induce the same path measure on $C([0,1],\mathcal M)$.
In particular, any intrinsic SB objective depending only on the induced path law and tangent control cannot distinguish between the two ambient representatives.
\end{proposition}
\begin{proof}
Let $\widetilde{u}_i(t,x) = u_{\parallel, i}(t,x) + u_{\perp, i}(t,x)$ denote the orthogonal decomposition of the ambient controls into tangential and normal components. The assumption $P_x \widetilde{u}_1(t,x) = P_x \widetilde{u}_2(t,x)$ implies $u_{\parallel, 1}(t,x) = u_{\parallel, 2}(t,x) =: u(t,x)$ for all $(t,x) \in [0,1]\times\mathcal M$. Substituting the ambient control into the drift term $f_t^{\mathcal{M}} + \sigma_t P_{X_t} \widetilde{u}_i$, we observe $P_{X_t} \widetilde{u}_i = P_{X_t}(u_{\parallel, i} + u_{\perp, i}) = u_{\parallel, i}$. Thus, the drift reduces to $f_t^{\mathcal{M}} + \sigma_t u_t$, which is independent of the normal component $u_{\perp, i}$. Consequently, the infinitesimal generators $\mathcal{L}_t^{(1)}$ and $\mathcal{L}_t^{(2)}$ are identical. By the well-posedness of the martingale problem \citep{ikeda1989stochastic, hsu2002stochastic} for diffusion processes on $\mathcal{M}$, the two SDEs induce the same path measure.
\end{proof}
Thus, standard NNs with outputs in $\mathbb R^d$ may be used, provided the output is projected before entering the dynamics and the control cost is evaluated on the projected field. This approach fully circumvents the need for explicit coordinate-based parameterization of the tangent bundle.

\subsection{Locality of the Tilted Heat-Kernel Approximation}
In the Euclidean zero-drift setting with no running cost, the flat heat semigroup satisfies $\nabla P_tf=P_t\nabla f$, because the Ricci curvature vanishes and all tangent spaces are canonically identified with $\mathbb R^d$. On a Riemannian manifold the exact replacement is Lemma~\ref{prop: exact_manifold_AM}: terminal gradients are propagated by the adjoint of damped stochastic parallel transport, which includes the Ricci-curvature correction in \eqref{eq: damped_transport}. Ordinary parallel transport along a chosen curve is therefore not an exact commutation rule, but the practical surrogate \eqref{eq: practical_transport_approximation}.
Since exact damped stochastic transport, exact manifold bridge sampling, and exact heat-kernel scores are generally unavailable in closed form, Algorithm~\ref{alg: asbs_m} replaces them with local geometric approximations.

Proposition~\ref{prop: tilted_endpoint_concentration} shows that, on a compact Riemannian manifold, tilting the heat kernel by a bounded terminal reward does not change the short-time Brownian scaling of the endpoint law: the tilted endpoint remains at geodesic distance\footnote{For a time-inhomogeneous noise schedule, $r$ should be interpreted as the effective diffusion time $r=\int_s^t \sigma_\tau^2\,d\tau$.} $O(\sqrt r)$ from its starting point with Gaussian tails.
Thus, for small effective diffusion time, the terminal reward tilt does not destroy the local nature of the short-time heat-kernel transition. Together with the fact that damped transport differs from ordinary parallel transport only through a bounded curvature-driven finite-variation term along the same short path, this supports the local substitutions in \eqref{eq: Varadhan_approximation} and \eqref{eq: practical_transport_approximation}. It does not, however, constitute a convergence proof of Algorithm~\ref{alg: asbs_m}, because the algorithm also replaces the stochastic reference bridge by a noisy geodesic construction.

\begin{proposition}[Endpoint concentration under the tilted heat-kernel law]
\label{prop: tilted_endpoint_concentration}
Let $\mathcal M$ be a compact Riemannian manifold\footnote{In particular, this applies to \(\mathbb T^n\), \(\mathbb S^n\), and
\(\mathrm{SO}(3)\) equipped with their standard metrics.} without boundary, and let $n:=\dim\mathcal M$. Let
$g^{\mathcal M}\in C^0(\mathcal M)$ and set $h(y):=e^{-g^{\mathcal M}(y)}$. For $r>0$, define
\begin{equation*}
  \varphi_r(x):=P_rh(x) = \int_{\mathcal M} h(y)p_r^{\mathcal M}(x,y) \,d\mathrm{vol}_{\mathcal M}(y),
\end{equation*}
where $P_r$ is the heat semigroup generated by $\frac12\Delta_{\mathcal M}$, and $p_r^{\mathcal M}$ is its scalar heat kernel.
Define the tilted endpoint law
\begin{equation*}
  p_{\mathcal M}^{\varphi}(\diff y\mid x) := \frac{h(y)p_r^{\mathcal M}(x,y)}{\varphi_r(x)} \,d\mathrm{vol}_{\mathcal M}(y).
\end{equation*}
Then there exist constants $C',C,c,r_0>0$, depending only on
$g^{\mathcal M}$ and $\mathcal M$, such that for all $x\in\mathcal M$,
all $0<r<r_0$, and all $\rho\ge0$,
\begin{equation*}
  p_{\mathcal M}^{\varphi} \bigl(d_{\mathcal M}(x,Y)\ge \rho \,\big|\, x \bigr) \le C\exp\left(-\frac{c\rho^2}{r}\right).
\end{equation*}
Consequently, $\mathbb E_{p_{\mathcal M}^{\varphi}(\cdot\mid x)} \left[ d_{\mathcal M}(x,Y)^2 \right] \le C' r$.
\end{proposition}

\begin{proof}
Since $\mathcal M$ is compact and $g^{\mathcal M}$ is continuous, there exist constants $0<h_-\le h_+<\infty$ such that $h_-\le h(y)\le h_+, \ y\in\mathcal M$. Because $p_r^{\mathcal M}(x,\cdot)$ is a probability density,
\begin{equation*}
  h_- \le \varphi_r(x) = \int_{\mathcal M}h(y)p_r^{\mathcal M}(x,y)\,d\mathrm{vol}_{\mathcal M}(y) \le h_+.
\end{equation*}
Hence $p_{\mathcal M}^{\varphi}(\diff y\mid x) \le \frac{h_+}{h_-} p_r^{\mathcal M}(x,y)\,d\mathrm{vol}_{\mathcal M}(y)$. 
The short-time Gaussian upper bound for the heat kernel on a compact
Riemannian manifold \citep{grigoryan1997gaussian, grigoryan2009heat} gives constants $C_0,c_0,r_0>0$ such that for $0<r<r_0$,
\begin{equation*}
  p_r^{\mathcal M}(x,y) \le C_0 r^{-n/2} \exp\left(-\frac{d_{\mathcal M}(x,y)^2}{c_0r}\right).
\end{equation*}
We now integrate this bound over the region $\{y:d_{\mathcal M}(x,y)\ge\rho\}$; the polynomial prefactor $r^{-n/2}$ is absorbed by an annulus decomposition. Set $c_1:=1/(2c_0)$ and $a:=\rho/\sqrt r$. If $a\le1$, the tail bound below holds trivially for any $C_1\ge e^{c_1}$, since the integral is at most $1$. If $a>1$, decompose the region into geodesic annuli $A_k:=\{y:\rho+k\sqrt r\le d_{\mathcal M}(x,y)<\rho+(k+1)\sqrt r\}$, $k\ge0$. 
Compactness yields a uniform volume-growth bound $\mathrm{vol}(B(x,s))\le C_Vs^n$ for all $x\in\mathcal M$ and $s>0$: for $s\le s_0$ by volume comparison under a lower Ricci bound \citep[Ch.~11]{lee2018introduction}, for $s\in[s_0,\operatorname{diam}\mathcal M]$ by $\mathrm{vol}(\mathcal M)/s_0^n$, and trivially for $s>\operatorname{diam}\mathcal M$ since then $B(x,s)=\mathcal M$; hence $\mathrm{vol}(A_k)\le C_V\bigl(\rho+(k+1)\sqrt r\bigr)^{n}=C_Vr^{n/2}(a+k+1)^{n}$. Therefore
\begin{align*}
\int_{\{y:d_\mathcal{M}(x,y)\ge \rho\}} p_r^{\mathcal M}(x,y)\,d\mathrm{vol}_{\mathcal M}(y)
&\le \sum_{k\ge0}\mathrm{vol}(A_k)\,\sup_{y\in A_k}p_r^{\mathcal M}(x,y) \\
&\le C_0C_V\sum_{k\ge0}\myellow{(a+k+1)^{n}}\,e^{-\frac{(a+k)^2}{c_0}}
\shortnote{heat-kernel and volume bounds}\\
&\le C_0C_V\myellow{\Bigl(\sup_{s\ge1}\,(2s)^{n}e^{-\frac{s^2}{2c_0}}\Bigr)}\sum_{k\ge0}e^{-\frac{(a+k)^2}{2c_0}}
\shortnote{$a{+}k{+}1\le2(a{+}k)$}\\
&\le C_1\,e^{-\frac{c_1\rho^2}{r}},
\shortnote{$(a{+}k)^2\ge a^2{+}k^2$}
\end{align*}
using $a+k+1\le2(a+k)$ for $a+k\ge1$, then $(a+k)^2\ge a^2+k^2$ to sum the geometric-type series; the constant $C_1$ depends only on $n$, $C_0$, $C_V$, and $c_0$. Combining with the tilt bound above,
\begin{equation*}
p_{\mathcal M}^{\varphi}\bigl(d_{\mathcal M}(x,Y)\ge\rho \,\big|\, x\bigr)
\le \frac{h_+}{h_-}\,C_1\,e^{-\frac{c_1\rho^2}{r}}
= C\,e^{-\frac{c\rho^2}{r}},
\end{equation*}
which is the claimed bound with $C:=(h_+/h_-)C_1$ and $c:=c_1$. The second-moment estimate follows from the tail identity (layer cake representation)
\begin{equation*}
\mathbb E[d_{\mathcal M}(x,Y)^2] = \int_0^\infty 2s\, p_{\mathcal M}^{\varphi} \bigl( d_{\mathcal M}(x,Y)\ge s \,\big|\, x \bigr) \,ds
\le \int_0^\infty 2s\,C e^{-\frac{cs^2}{r}}\,ds
= \frac{C}{c}\,r =: C' r.
\end{equation*}
\end{proof}

\subsection{Computational Complexity}
We summarize the per-epoch cost of R--ASBS (Algorithm~\ref{alg: asbs_m}) and Extended R--ASBS (Algorithm~\ref{alg: extended_asbs_m}). Let \(B\) be the batch size, \(N\) the number of SDE steps, and \(K_{\mathrm{tr}}\) the number of training epochs. Let \(F_u(M)\), \(F_h(M)\) denote batched forward evaluations of the controller and corrector on \(M\) inputs, and let \(G_u(M)\), \(G_h(M)\) denote the corresponding gradient-update costs. Let \(P_B\), \(R_B\), and \(\Pi_B\) denote, respectively, the cost of applying tangent projection, retraction, and nearest-point projection to a batch of \(B\) particles. For a general embedded equality-constraint manifold,
\[
P_B
=
O\!\left(B(C_{J_c}+m^2d+m^3)\right),
\]
using dense linear algebra for
\[
P_x=I-J_c(x)^\top(J_c(x)J_c(x)^\top)^{-1}J_c(x).
\]
For structured manifolds such as spheres or tori, these operations reduce to closed-form \(O(Bd)\) primitives.

For R--ASBS, each epoch performs two controlled rollouts, one before and one after the controller update. Each rollout requires a controller evaluation, tangent projection of both control and noise, and retraction at every SDE step. Thus the rollout cost is
\[
2N\bigl(F_u(B)+2P_B+R_B\bigr).
\]
The remaining operations, namely bridge interpolation, terminal adjoint evaluation, vector transport, endpoint-score/corrector-target computation, and the controller/corrector updates, are performed on \(B\) samples. If \(C_{\mathrm{geom},B}\) denotes the combined cost of these geometric operations, then
\[
T_{\mathrm{R\text{-}ASBS}}
=
O\!\left(
2N(F_u(B)+2P_B+R_B)
+
F_h(B)+G_u(B)+G_h(B)+C_{\mathrm{geom},B}
\right).
\]
With fixed network size and closed-form geometric primitives, \(T_{\mathrm{R\text{-}ASBS}}=O(BN)\).

For Extended R--ASBS, each epoch starts with an initial projection of the source particles, costing \(\Pi_B\). It then performs two projected rollouts, each costing
\[
N(F_u(B)+2P_B+\Pi_B).
\]
Projection-as-transport applies tangent projection along the stored trajectory and costs \(NP_B\). The controller loss is evaluated on all \(BN\) trajectory states, while the corrector loss is evaluated only on the \(B\) terminal states. Therefore,
\[
T_{\mathrm{Ext}}
=
O\!\left(
\Pi_B
+
2N(F_u(B)+2P_B+\Pi_B)
+
NP_B
+
F_h(B)+G_u(BN)+G_h(B)+C_{E,B}+P_B
\right),
\]
where \(C_{E,B}\) is the cost of evaluating the terminal energy gradients. For fixed network size and a bounded number of Newton iterations in \(\Pi_M\), this also scales as \(O(BN)\), but with a larger constant due to repeated nearest-point projections, projection-as-transport, and the \(BN\)-sample controller update.

Let $N_u := |\theta_u|,\ N_h := |\theta_h|$ be the numbers of network parameters. 
The memory cost of R--ASBS is
\[
M_{\mathrm{R\text{-}ASBS}}=O(Bd+N_u+N_h),
\]
up to neural-network activation memory, since rollouts can be streamed and only endpoints and bridge intermediates are retained. Extended R--ASBS stores the full trajectory and transported adjoints, hence
\[
M_{\mathrm{Ext}}=O(BNd+N_u+N_h),
\]
again plus activation memory for the controller update on \(BN\) inputs.

After training, both methods generate samples by simulating only the learned controlled diffusion. The corrector is not used at sampling time. For \(S\) generated samples,
\[
T_{\mathrm{gen}}^{\mathrm{R\text{-}ASBS}}
=
O\!\left(N(F_u(S)+2P_S+R_S)\right),
\qquad
T_{\mathrm{gen}}^{\mathrm{Ext}}
=
O\!\left(N(F_u(S)+2P_S+\Pi_S)\right).
\]
Thus, training is the expensive phase, while post-training sampling is amortized and linear in both the number of samples and SDE steps.

\section{Additional Numerical Results}
\label{app: additional_num_results}
\subsection{Stiefel manifolds}
\label{app: app_stiefel_manifold}
The Stiefel manifold is defined as $St(n,p) := \{ X \in \R^{n \times p}: X^\T X = I_p\}$. Differentiating the constraint $c(X) = X^\T X - I_p = 0$ yields the tangent space at $X$
\begin{equation}
    T_X St(n,p) = \{V \in \R^{n \times p}: V^\T X + X^\T V = 0\}.
    \label{eq: TX_St}
\end{equation}
To project an arbitrary ambient matrix $Z \in \R^{n \times p}$ onto \eqref{eq: TX_St}, we decompose $Z$ into a tangent component $V$ and a normal component. The normal space at $X$ is given by $N_X St(n,p) = \{XS: S=S^\T \in \R^{p \times p} \}$.
Setting $Z=V +XS$ and left-multiplying by $X^\T$, we get $X^\T Z = X^\T V + S$. Because $V \in T_X St(n,p)$, $X^\T V$ is a skew-symmetric matrix.
Taking the symmetric part of both sides isolates $S$, yielding $S = \mathrm{sym}(X^\T Z) := \frac{1}{2} (X^\T Z + Z^\T X)$.
Therefore, the orthogonal projector $P_X : \R^{n \times p} \to T_X St(n,p)$ is
\begin{equation}
    P_X(Z) = Z - X \frac{X^\T Z + Z^\T X}{2}.
    \label{eq: PX_St}
\end{equation}
Using \eqref{eq: PX_St}, the Riemannian gradient is simply the orthogonal projection of the ambient Euclidean gradient
\begin{equation*}
    \nabla_\mathcal{M} E(X) = P_X(\nabla_{\R^{n \times p}} E(X))
    = \nabla_{\R^{n \times p}} E(X) - X \mathrm{sym}\big(X^\T \nabla_{\R^{n \times p}} E(X)\big).
\end{equation*}
Computing the exact exponential map on the Stiefel manifold is computationally expensive as it involves computing matrix exponentials of size $2p \times 2p$ or solving ODEs. Therefore, 
We use the QR retraction map\footnote{A second-order alternative is the polar retraction $R_X^{\mathrm{polar}}(V) =(X+V)\bigl((X+V)^\top(X+V)\bigr)^{-1/2}$.} $R_X : T_X St (n,p) \to St(n,p)$ defined as $R_X(V) = \mathrm{qf}(X + V)$, where $\mathrm{qf}(\cdot)$ extracts the orthogonal $Q$ factor from the QR decomposition.
The absence of an analytical expression for a transport map led to the implementation of Extended R--ASBS.

For the experiment in Figure~\ref{fig: stiefel_manifold}, we use
$H = \begin{bmatrix}
        A & B \\
        B & A
    \end{bmatrix}, \ \substack{p=2\\ n=4}$,
$A := \begin{bmatrix}
        4 & 0.5\\
        0.5 & 4
    \end{bmatrix},
B := \begin{bmatrix}
       2.5 & 1 \\
       1 & 2.5 
    \end{bmatrix}.$

\subsection{Robust Wahba problem}
\label{app: app_robust_Wahba_problem}
\begin{figure}[!t]
\centering
\includegraphics[width=0.50\columnwidth]{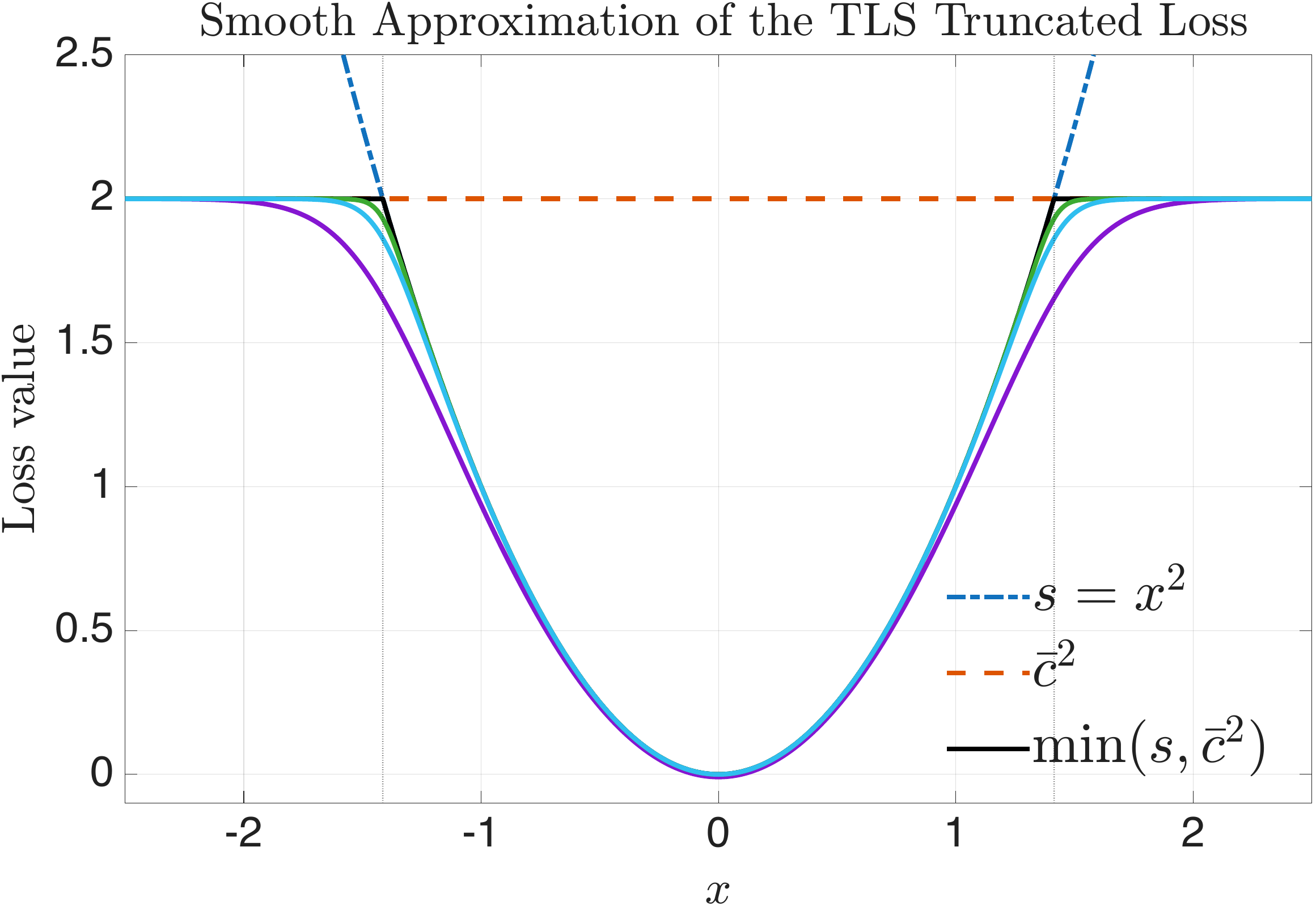}
\captionof{figure}{The figure illustrates smoothing for $\tau \in \{0.1, 0.2, 0.5\}$.}
\label{fig: smoothed_loss}
\end{figure}
During neural network training, the non-differentiability of the energy function $\min(s,\bar{c}^2)$ is avoided by employing the smooth approximation $-\tau \log (e^{-s/\tau} + e^{-\bar{c}^2 / \tau})$ with $\tau \ll 1$ (Figure \ref{fig: smoothed_loss}).
\begin{figure}[!t]
\centering
\includegraphics[width=0.70\columnwidth]{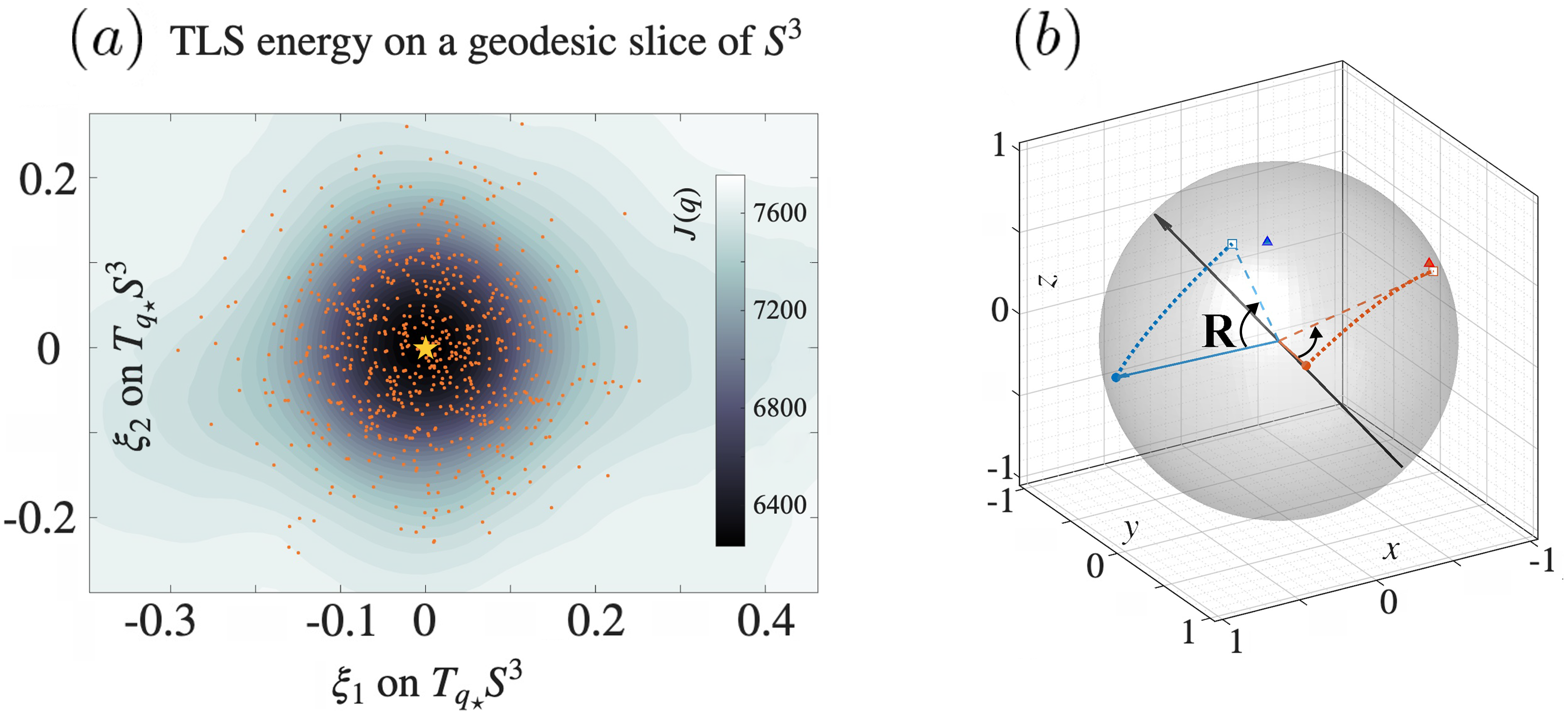}
\caption{\( (a) \) TLS objective evaluated on a geodesic slice of \( \mathbb{S}^3 \) in a neighborhood of \( \boldsymbol{q}^\star \), with Extended R--ASBS samples (\( \textcolor{orange}{\bullet} \)) concentrated near the optimal quaternion ($\color{yellow!70!black} \star$). $(b)$ Directional correspondences on \( \mathbb S^2 \), showing source vectors ($\textcolor{cyan}{\bullet}, \textcolor{orange}{\bullet}$), true rotated directions (dashed lines), measured endpoints affected by zero-mean Gaussian noise ($\textcolor{cyan}{\blacktriangle}, \textcolor{orange}{\blacktriangle}$), rotation arcs (dot lines), and the true rotation axis (black arrow).}
\label{fig: wahba_problem}
\end{figure}
The ablation study in Table~\ref{tab: asbsm_wahba_extra_results} shows that the sample budget is the dominant factor affecting performance (Panel~A), whereas R--ASBS remains comparatively robust to the choice of $\beta$ over the tested range (Panel~B).
\begin{table}[H]
\centering
\caption{R--ASBS parameter sensitivity for robust Wahba on $\mathrm{SO}(3)$ via $\mathbb{S}^3$ quaternions.}
\label{tab: asbsm_wahba_extra_results}
\small
\newcommand{\sd}[1]{\graysd{#1}}

\begin{minipage}[t]{0.48\textwidth}
\centering
\renewcommand{\arraystretch}{1.2}
\setlength{\tabcolsep}{7.0pt}
\begin{tabular}{c c c}
\toprule
\multicolumn{3}{c}{Panel A: Effect of $N_s$, $\beta=1$} \\
\midrule
$N_s$ &
TLS gap &
R--ASBS err. \\
&
(\%) &
(deg) \\
\midrule
$500$
& $0.290$ \sd{0.504}
& $2.055$ \sd{0.991} \\

$1000$
& $0.131$ \sd{0.193}
& $1.662$ \sd{0.822} \\

$2500$
& $0.054$ \sd{0.131}
& $1.290$ \sd{0.679} \\

$5000$
& $0.033$ \sd{0.117}
& $1.269$ \sd{0.735} \\

$10000$
& \cellhi $0.003$ \sd{0.090}
& \cellhi $1.175$ \sd{0.750} \\
\bottomrule
\end{tabular}
\end{minipage}
\hfill
\begin{minipage}[t]{0.48\textwidth}
\centering
\renewcommand{\arraystretch}{1.2}
\setlength{\tabcolsep}{7.0pt}
\begin{tabular}{c c c}
\toprule
\multicolumn{3}{c}{Panel B: Effect of $\beta$, $N_s=10{,}000$} \\
\midrule
$\beta$ &
TLS gap &
R--ASBS err. \\
&
(\%) &
(deg) \\
\midrule
$0.25$
& $-0.015$ \sd{0.097}
& $1.540$ \sd{0.917} \\

$0.50$
& $-0.010$ \sd{0.080}
& $1.235$ \sd{0.701} \\

$1$
& \cellhi $0.003$ \sd{0.090}
& $1.175$ \sd{0.750} \\

$2$
& $0.025$ \sd{0.107}
& \cellhi $1.169$ \sd{0.615} \\

$4$
& $-0.015$ \sd{0.081}
& $1.371$ \sd{0.929} \\
\bottomrule
\end{tabular}
\end{minipage}
\end{table}
All simulations in Section  \ref{sec: numerical_resuls} were implemented in MATLAB R2025b using the Deep Learning Toolbox. The source code is available at: \href{https://github.com/mattiamosso/Hard-Constrained-Sampling-on-Embedded-Riemannian-Manifold-via-Adjoint-Schrodinger-Bridges}{GitHub repository}.
\end{document}